\documentclass[11pt]{article}

\usepackage{geometry}
\usepackage{url}
\usepackage{graphicx,graphics}
\usepackage{subcaption}
\usepackage{amsmath,amssymb,latexsym,lscape,amsthm,mathrsfs}
\usepackage{xcolor}
\usepackage{lineno,cite}
\usepackage[unicode=true,pdfusetitle,
bookmarks=true,bookmarksnumbered=false,bookmarksopen=false,breaklinks=false,pdfborder={0 0 0},backref=false,colorlinks=false]
{hyperref}
\theoremstyle{plain}

\newtheorem{prop}{Proposition}
\newtheorem{lemma}{Lemma}
\theoremstyle{definition}
\newtheorem{definition}{Definition}
\newtheorem{remark}{Remark}
\usepackage{tikz}
\usepackage{adjustbox}
\usepackage{appendix}
\usepackage{mathrsfs}
\usepackage{array}
\usepackage{booktabs}
\graphicspath{{./}}
\definecolor{blue}{rgb}{0.0, 0.0, 1.0}
\definecolor{bleudefrance}{rgb}{0.19, 0.55, 0.98}
\definecolor{deepskyblue}{rgb}{0.0, 0.75, 1.0}
\definecolor{green}{rgb}{0.0, 0.5, 0.0}
\definecolor{darkblue}{rgb}{0.4, 0.70, 1.33}
\hypersetup{colorlinks,linkcolor={black},citecolor={green},urlcolor={black}}
\usepackage{comment}
\usepackage{cleveref}
\usepackage{braket}
\usepackage{mathabx}
\usepackage{epstopdf}

\newcommand{\R}{\mathbb{R}}

\newcommand{\Rzero}{\mathcal{R}_0}
\newcommand{\Rc}{\mathcal{R}_C}
\newcommand{\dist}{\emph{dist}}

\definecolor{VaccinateColor}{HTML}{99CBFF}

\usepackage{booktabs}
\usepackage{adjustbox} 

\usepackage[affil-it]{authblk} 
\usepackage[affil-it]{authblk} 


\title{Impact of Nirsevimab prophylaxis on RSV dynamics: a stage-structured modelling study}

\author[1,2]{Anna Autoriello}
\author[1,2]{Sabrina Averga}
\author[2]{Bruno Buonomo\thanks{Corresponding author:
\href{mailto:buonomo@unina.it}{buonomo@unina.it}.}}
\author[2]{Rossella Della Marca}
\author[3,4]{Alfredo Guarino\protect\footnotemark[1]}
\author[3,4]{Andrea Lo Vecchio}
\author[1,3]{Cristina Moracas}
\author[2]{Emanuela Penitente}
\author[3,4]{Marco Poeta}

\affil[1]{PhD National Programme in One Health approaches to infectious diseases and life science research, Department of Public Health, Experimental and Forensic Medicine, University of Pavia, viale Golgi 19, 27100 Pavia, Italy}
\affil[2]{Department of Mathematics and Applications, University of Naples Federico II, via Cintia, 80126 Naples, Italy}
\affil[3]{Pediatric Infectious Disease Unit, Department of Maternal and Child Health, University Hospital Federico II, via Pansini 5, 80131 Naples, Italy}
\affil[4]{Department of Translational Medical Science, University of Naples Federico II, via Pansini 5, 80131 Naples, Italy}

\date{}

\begin{document}

\maketitle

\vspace*{-30pt}

\begin{abstract}
Respiratory syncytial virus (RSV) is a leading cause of bronchiolitis and other lower respiratory tract infections in infants. Increased viral circulation in the post-COVID era and heterogeneous prevention strategies across regions have made RSV control more challenging. We develop a stage-structured, age-stratified Susceptible--Infected--Recovered (SIR) compartmental model tailored to the Italian setting to investigate the population-level impact of infant prophylaxis with Nirsevimab, a long-acting monoclonal antibody. Scenario-based simulations over a multi-year horizon show that increasing infant protection coverage substantially reduces RSV incidence among infants and also yields indirect benefits in older age groups. In particular, extending coverage to infants born outside the epidemic season further lowers cumulative incidence, although infant-targeted prophylaxis alone does not reduce the control reproduction number below the epidemic threshold in the parameter range explored. These findings suggest that broader and more consistent infant Nirsevimab coverage may reduce RSV burden and support the evaluation of alternative implementation strategies in the Italian context.
\end{abstract}

\medskip

\noindent { {\bf Keywords}: mathematical epidemiology, respiratory syncytial virus, stage-structured model, infant immunoprophylaxis, Nirsevimab.}
\normalsize

\section{Introduction}

Respiratory syncytial virus (RSV) is one of the leading causes of hospitalisation for bronchiolitis and other lower respiratory tract infections in infants, representing a major cause of pressure on paediatric healthcare systems~\cite{manzoni2025prevention,o2019causes}. With the introduction of new preventive interventions, especially long--acting monoclonal antibodies for infants, an urgent question is how real--world delivery choices translate into population-level impact~\cite{lopez2024potential}. Recent post--pandemic seasons, for instance in Italy, have highlighted how changes in respiratory virus circulation~\cite{iss_ili} and heterogeneous prevention policies~\cite{SIP,Corriere} can complicate planning based on surveillance alone. In this context, there is a clear need for mathematical tools that complement clinical and public health evidence by translating intervention assumptions into quantitative, population-level expectations and by enabling transparent comparisons across alternative implementation scenarios \cite{lang2022use}.

RSV is a major cause of acute respiratory infections worldwide and affects individuals across all age groups, with clinical manifestations ranging from mild to severe disease~\cite{WHO}. Although RSV infection occurs throughout the life course, susceptibility to severe disease is strongly age-dependent, being highest in early childhood and later adulthood. Furthermore, naturally acquired immunity provides only limited and short-term protection, resulting in frequent reinfections~\cite{manzoni2025prevention}.

Worldwide, RSV is associated each year with approximately 3.6 million hospital admissions and around 100{,}000 deaths among children under five years of age, with the highest mortality (nearly 50\%) observed during the first six months of life and in low- and middle-income countries~\cite{WHO}. Beyond childhood, RSV is increasingly recognised as an important cause of severe respiratory illness in older adults as well, particularly among individuals with underlying chronic conditions, and contributes substantially to hospitalisations and mortality in high-income settings~\cite{WHO}.

In Europe, RSV accounts annually for a considerable number of hospital admissions in both paediatric and adult populations. Estimates indicate that more than 200{,}000 hospital admissions occur annually among children under five years of age, with marked seasonal peaks during the winter months~\cite{PortaleEuropeo}. Among the adult population, European estimates indicate approximately 160{,}000 RSV-associated hospital admissions per year, with about 92\% of these occurring in individuals aged 65 years and older~\cite{PortaleEuropeo,osei2023estimation}. In Italy, RSV is estimated to cause around 26,000 hospitalisations annually among adults aged 60 years and above, resulting in approximately 1,800 deaths~\cite{manzoni2025prevention}.

Given the significant burden of RSV infection among infants, preventive strategies have improved substantially in recent years. For more than two decades, the humanised monoclonal antibody Palivizumab has been used to prevent severe RSV disease in preterm and high-risk infants~\cite{impact1998palivizumab}. However, its high cost and short half-life, requiring monthly administration during the RSV season, have limited its use to these high-risk groups~\cite{jha2016respiratory}. Currently, the primary preventive approach in infants relies on long-acting monoclonal antibodies \cite{WHO}.  Nirsevimab, a monoclonal antibody, is currently recommended for all infants younger than 8 months entering their first RSV season, or born during the epidemic period~\cite{CDCanticorpimonoclonali}. It is administered as a single intramuscular dose and provides immediate passive protection lasting approximately five months, covering an entire RSV season~\cite{nobili2024strategie}. Maternal immunisation has also been introduced as an additional option for infant protection, although its implementation remains limited in several European countries, including Italy~\cite{du2025impact,FondazioneVeronesistrategiedivaccinazione}. RSV vaccines for older adults have also become available as an additional preventive tool, but estimates of protection duration and effectiveness in routine use are still uncertain \cite{cdc2025vaccineforadults}.

In this context, mathematical modelling provides a complementary framework for investigating the transmission dynamics of infectious diseases and for gaining insight into the potential effects of intervention strategies \cite{hethcote2000mathematics}. In particular, compartmental models based on systems of ordinary differential equations are especially valuable in settings where surveillance data are incomplete, heterogeneous, or insufficient to support precise quantitative estimates, but still informative for exploring qualitative dynamics and comparative scenarios \cite{overton2020using,keeling2008modeling}. This is the case for respiratory syncytial virus, for which uncertainties in reporting, underdiagnosis, and limited age-specific data constrain purely data-driven approaches.

Over the past two decades, a variety of mathematical models have been developed to study RSV transmission and control. Acedo et al.~\cite{acedo2010mathematical} propose an age-stratified Susceptible--Infected--Recovered--Susceptible (SIRS) model with two groups, distinguishing infants in the first year of life from the rest of the population. Some modelling approaches have focused on immune dynamics rather than age stratification, introducing multi-stage compartmental formulations in which successive stages represent repeated RSV infections~\cite{weber2001modeling,greenhalgh2009backward} and, in some cases, incorporating the seasonal nature of the disease~\cite{weber2001modeling,arenas2008existence}.

More recent models have increasingly incorporated a stage-structured formulation stratified by age (SSAS models) to reflect the strong heterogeneity in susceptibility and infection dynamics across paediatric and adult populations~\cite{moore2014modelling,hogan2016exploring}. This modelling choice is further supported by recent epidemiological evidence showing that the severity of RSV-related disease is heavily concentrated in the first year of life, with infants accounting for the majority of severe cases~\cite{menegale2025impact}. More detailed stratifications have also been explored, including models with finely resolved age classes and seasonally varying transmission rates~\cite{hogan2017potential}. Other approaches, such as individual-based models~\cite{poletti2015evaluating}, explicitly represent contact structures within households and communities. While these approaches offer more detailed insights, they typically require extensive parametrisation and high-quality data.

SSAS models have also been widely applied to other infectious diseases, including measles, HIV/AIDS, tuberculosis, and respiratory infections in general~\cite{hethcote1997age,aldila2018mathematical,zhao2020dynamic,lee2021age,chen2024dynamical}. In particular, Zhao et al.~\cite{zhao2020dynamic} combine age stratification with nonconstant population dynamics and provides a methodological reference closely aligned with the model structure adopted in the present study.

In this paper, we develop a stage-structured SIR-type model stratified by age to investigate RSV transmission and prevention in the Italian context. The population is divided into three age classes reflecting the strong age dependence observed in RSV epidemiology, with particular attention to infancy and older adulthood. Consistent with current preventive practices and data availability, the model explicitly incorporates infant immunoprophylaxis with the long-acting monoclonal antibody Nirsevimab.  This framework is used to derive the basic and control reproduction numbers and to run multi-year scenario-based simulations that quantify how alternative assumptions of Nirsevimab delivery, at birth and through catch-up administration, affect annual incidence and its age distribution over consecutive years.

The remainder of the paper is organised as follows. Section~\ref{sec:model} presents the formulation of the mathematical model and outlines its main structural assumptions. Section~\ref{sec:rep_num} introduces the reproduction numbers, discussing their epidemiological meaning and their role in characterising RSV transmission dynamics. Section~\ref{sec:param} describes the model parametrisation, and Section~\ref{sec:num_res} reports the numerical results concerning the exploration of different intervention scenarios and the assessment of their impact on key epidemiological outcomes. Finally, Section~\ref{sec:conclusion} discusses the main findings and their implications for RSV prevention in the proposed modelling framework.

\section{Model formulation}\label{sec:model}
\subsection{Age structure and compartments}
We consider a population stratified into three age groups: infants ($<$1 year), children--and--adults (1-64 years), and seniors ($\geq 65$ years). Within each group, the transmission of respiratory syncytial virus (RSV) is modelled using a SIR--type compartmental framework, where individuals are classified as susceptible ($S_i$), infected ($I_i$), and recovered ($R_i$), for $i = 1,2,3$. 

The choice of a SIR structure is motivated by the aim of keeping the model analytically tractable while still capturing the essential features of RSV transmission. In particular, several models developed for RSV neglect an exposed class and assume direct progression from susceptibility to infection \cite{greenhalgh2000subcritical,greenhalgh2009backward}. Several studies have also reported similarities between bovine and human RSV in terms of genomic structure, transmission, and clinical features \cite{da2024interactions,baker1991human,borchers2013respiratory,taylor2013bovine}. Hence, although age-structured SEIR models with explicit latent stages have been proposed \cite{moore2014modelling,zhang2012existence,weber2001modeling,hogan2016exploring}, the use of a SIR framework remains consistent with the available evidence and allows for a more tractable model structure.

\subsection{Protection dynamics}
In the infant group, we account for the administration of Nirsevimab by introducing additional compartments. We denote by $P$ the class of infants who are protected by Nirsevimab prophylaxis, and divide the susceptibles into two compartments: $S_{1e}$, comprising infants eligible for mAb prophylaxis who have not yet received it, and $S_{1w}$, comprising susceptible infants whose antibody-derived protection or natural immunity has waned.
We assume that Nirsevimab is administered at birth to a fraction $p$ of newborns; therefore, this fraction directly enters the protected compartment, $P$. The remaining $1-p$ newborns enter the $S_{1e}$ compartment, from which they may subsequently receive mAb prophylaxis at rate $\psi$. The protected infants are assumed to have partial protection, characterised by an ineffectiveness factor $\sigma \in (0,1)$, and to lose protection at rate $\omega$, i.e. the inverse of the average duration of protection. Infants whose antibody-derived protection wanes over time re-enter susceptibility. A fraction $q$ of them is still under one year of age and is placed in the $S_{1w}$ compartment; the remaining 
$1-q$ is over one year and, therefore, moves to the $S_2$ compartment. Consequently, infants are subdivided into five disjoint compartments: $P$, $S_{1e}$, $S_{1w}$, $I_1$, and $R_1$. 
For children--and--adults and senior individuals (age groups 2 and 3), we adopt the standard SIR structure. 

Although prevention of RSV in older adults is epidemiologically relevant, it is not explicitly included in the present model, as the primary aim of this study is to investigate the population-level impact of infant prophylaxis with Nirsevimab. This choice is also consistent with the Italian setting considered in this study, where RSV imposes a documented burden among Italian older adults and high-risk individuals \cite{DomnichCalabro2024,Puggina2025} but evidence on RSV vaccination in older adults remains recent and limited, with available data mainly referring to early adopters in specific regional settings \cite{Domnich2025}.

\subsection{Demography, ageing and deaths}
Let $S_i(t)$ denote the number of susceptible individuals in age group $i=2,3$ at time t. Let $I_i(t)$ and $R_i(t)$ denote the number of infectious and recovered individuals in age group $i=1,2,3$ at time t, respectively. Furthermore, $S_{1e}(t)$ denotes the number of susceptible infants eligible for mAb prophylaxis at time t, $S_{1w}(t)$ the number of susceptible infants whose protection or immunity has waned at time t, and $P(t)$ the number of protected infants at time t. The total population in each age group at time t is given by:
\[
N_1(t) = P(t) + S_{1e}(t)+S_{1w}(t) + I_1(t) + R_1(t),
\]
\[\quad N_i(t) = S_i(t) + I_i(t) + R_i(t), \text{  for } i=2,3,
\]
and the total population size at time t is $N(t) = N_1(t) + N_2(t) + N_3(t)$.

Individuals move from $i$--th age group to the $(i+1)$--th age group at a rate $\eta_i$, $i=1,2$, corresponding to the inverse of the average timespan spent in the $i$-th age group. We assume that individuals retain their disease status upon ageing.

The model accounts for demographic turnover through births, represented by a net inflow $\Lambda$ into the infant age groups $P$ and $S_{1e}$, and through natural mortality, represented by the age specific rates $\mu_i$, $i=1,2,3$; this turnover is essential to represent interventions delivered at birth.
Due to the short timespan of our analysis, we assume that net inflow $\Lambda$ is constant  (see, e.g., \cite{acedo2010mathematical,moore2014modelling,hogan2016exploring,hogan2017potential,hodgson2020evaluating}).

In addition to natural mortality, the model includes disease-induced mortality in the first and third age groups, occurring at rates $d_1$ and $d_3$, respectively. Disease-induced mortality is not considered in the children-and-adult group, as empirical evidence indicates that it primarily affects high-risk individuals, who represent only a small fraction of the total adult population \cite{ECDC_RSV_RRA_2022}.

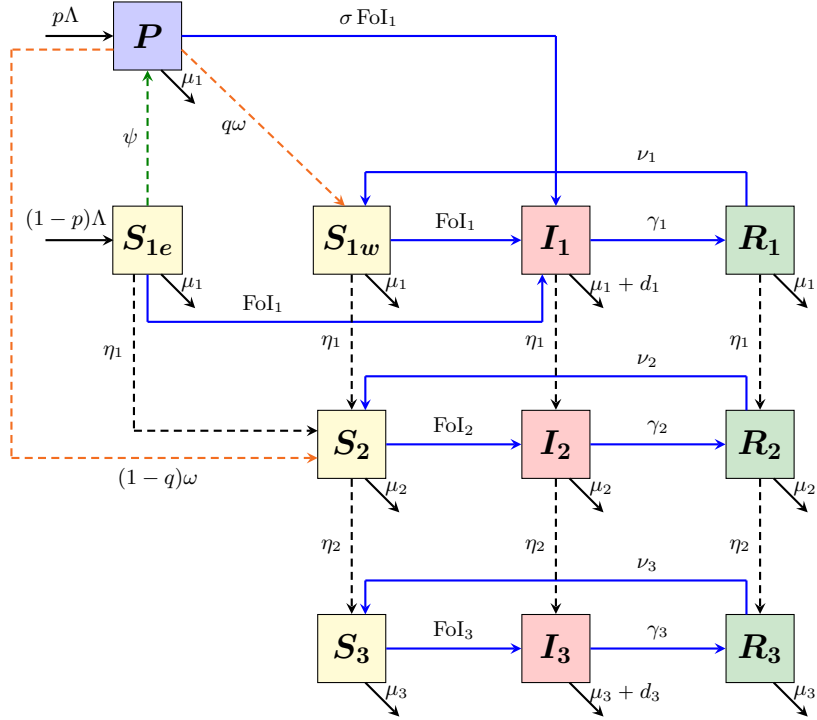
\begin{figure}
    \tikzstyle{P} = [
        rectangle, minimum width=1cm, minimum height=1cm,
        text centered, draw=black,fill=blue!20
    ]
    \tikzstyle{S} = [
        rectangle, minimum width=1cm, minimum height=1cm,
        text centered, draw=black,fill=yellow!20
    ]
    \tikzstyle{I} = [
        rectangle, minimum width=1cm, minimum height=1cm,
        text centered, draw=black,fill=red!20
    ]
    \tikzstyle{R} = [
        rectangle, minimum width=1cm, minimum height=1cm,
        text centered, draw=black,fill=green!20
    ]
    \tikzstyle{arrow} = [thick,->,>=stealth,blue]
    \tikzstyle{arrowfit} = [thick,-,blue]
    \tikzstyle{dashedarrow} = [arrow, densely dashed,black]
    \tikzstyle{dashedfitarrow} = [thick,-,densely dashed,black]
    \tikzstyle{greenarrow} = [arrow, green,densely dashed]
    \tikzstyle{blackarrow} = [arrow,black]
    \tikzstyle{redarrow} = [arrow,red,densely dashed]
    \tikzstyle{arrowfitred} = [thick,-,red,densely dashed]
    \tikzstyle{orangearrow} = [arrow, yellow!30!red, densely dashed]
    \tikzstyle{orangefitarrow} = [thick,-,yellow!30!red, densely dashed]
    
    \centering
    \begin{tikzpicture}[node distance=3cm, scale=0.9, transform shape]

        \node (I1) [I]  {\Large $\boldsymbol{I_1}$};
        \node (S1w) [S, left of=I1]{\Large $\boldsymbol{S_{1w}}$};
        \node (S1e) [S, left of=S1w]{\Large $\boldsymbol{S_{1e}}$};
        \node (R1) [R, right of=I1]{\Large $\boldsymbol{R_1}$};
        \node (P) [P, above of=S1e]{\Large $\boldsymbol{P}$};
        \node (S2) [S, below of=S1w]{\Large $\boldsymbol{S_2}$};
        \node (I2) [I, right of=S2]{\Large $\boldsymbol{I_2}$};
        \node (R2) [R, right of=I2]{\Large $\boldsymbol{R_2}$};
        \node (S3) [S, below of=S2]{\Large $\boldsymbol{S_3}$};
        \node (I3) [I, right of=S3]{\Large $\boldsymbol{I_3}$};
        \node (R3) [R, right of=I3]{\Large $\boldsymbol{R_3}$};

        \draw [blackarrow] (-5.8,2.5) -- node[pos =0.3, right] {\footnotesize{$\mu_1$}}(-5.3,2);

        \draw [blackarrow] (-2.8,-0.5) -- node[pos =0.3, right] {\footnotesize{$\mu_1$}}(-2.3,-1);

        \draw [blackarrow] (-5.8,-0.5) -- node[pos =0.3, right] {\footnotesize{$\mu_1$}}(-5.3,-1);
        
        \draw [blackarrow] (0.2,-0.5) -- node[pos =0.3, right] {\footnotesize{$\mu_1+d_1$}}(0.7,-1);
        
        \draw [blackarrow] (3.2,-0.5) -- node[pos =0.3, right] {\footnotesize{$\mu_1$}}(3.7,-1);

        \draw [blackarrow] (-2.8,-3.5) -- node[pos =0.3, right] {\footnotesize{$\mu_2$}}(-2.3,-4);
        
        \draw [blackarrow] (0.2,-3.5) -- node[pos =0.3, right] {\footnotesize{$\mu_2$}}(0.7,-4);
        
        \draw [blackarrow] (3.2,-3.5) -- node[pos =0.3, right] {\footnotesize{$\mu_2$}}(3.7,-4);

        \draw [blackarrow] (-2.8,-6.5) -- node[pos =0.3, right] {\footnotesize{$\mu_3$}}(-2.3,-7);
        
        \draw [blackarrow] (0.2,-6.5) -- node[pos =0.3, right] {\footnotesize{$\mu_3+d_3$}}(0.7,-7);
        
        \draw [blackarrow] (3.2,-6.5) -- node[pos =0.3, right] {\footnotesize{$\mu_3$}}(3.7,-7);
        
        \draw [arrow] (S1w) -- node[above] {\footnotesize{\textcolor{black}{$\mathrm{FoI}_1$}}}(I1);
        \draw [arrow] (S2) -- node[above] {\footnotesize{\textcolor{black}{$\mathrm{FoI}_2$}}}(I2);
        \draw [arrow] (S3) -- node[above] {\footnotesize{\textcolor{black}{$\mathrm{FoI}_3$}}}(I3);
        \draw[arrowfit](-5.5,3) -- node[above]{\textcolor{black}{\footnotesize{$\sigma\,\mathrm{FoI}_1$}}}(0,3);
        \draw[arrow] (0,3) -- (0,0.5);
        \draw[arrowfit] (S1e.south) -- (-6,-1.2);
        \draw[arrowfit]  (-6,-1.2) -- node[above, xshift=-1.2cm] {\footnotesize{\textcolor{black}{$ \mathrm{FoI}_1$}}}(-0.2,-1.2);
        \draw[arrow]  (-0.2,-1.2) -- (-0.2,-0.5);
        
        \draw [arrow] (I1) -- node[above] {\footnotesize{\textcolor{black}{$\gamma_1$}}}(R1);
        \draw [arrow] (I2) -- node[above] {\footnotesize{\textcolor{black}{$\gamma_2$}}}(R2);
        \draw [arrow] (I3) -- node[above] {\footnotesize{\textcolor{black}{$\gamma_3$}}}(R3);
        
        \draw [arrowfit] (2.8,0.5) -- (2.8,1);
        \draw [arrowfit] (2.8,1) -- node[pos=0.26, above] {\footnotesize{\textcolor{black}{$\nu_1$}}} (-2.8,1);
        \draw [arrow] (-2.8,1) -- (-2.8,0.5);
        \draw [arrowfit] (2.8,-2.5) -- (2.8,-2);
        \draw [arrowfit] (2.8,-2) -- node[pos=0.26, above] {\footnotesize{\textcolor{black}{$\nu_2$}}} (-2.8,-2);
        \draw [arrow] (-2.8,-2) -- (-2.8,-2.5);
        \draw [arrowfit] (2.8,-5.5) -- (2.8,-5);
        \draw [arrowfit] (2.8,-5) -- node[pos=0.26, above] {\footnotesize{\textcolor{black}{$\nu_3$}}} (-2.8,-5);
        \draw [arrow] (-2.8,-5) -- (-2.8,-5.5);
        
        \draw[dashedarrow] (S1w) -- node[left] {\footnotesize{$\eta_1$}}(S2);
        \draw[dashedarrow] (I1) -- node[left] {\footnotesize{$\eta_1$}}(I2);
        \draw[dashedarrow] (R1) -- node[left] {\footnotesize{$\eta_1$}}(R2);
        \draw[dashedarrow] (S2) -- node[left] {\footnotesize{$\eta_2$}}(S3);
        \draw[dashedarrow] (I2) -- node[left] {\footnotesize{$\eta_2$}}(I3);
        \draw[dashedarrow] (R2) -- node[left] {\footnotesize{$\eta_2$}}(R3);
        \draw[dashedfitarrow] (-6.2,-0.5) -- node[left] {\footnotesize{\textcolor{black}{$\eta_1$}}} (-6.2,-2.8); 
        \draw[dashedarrow] (-6.2,-2.8) -- (-3.5,-2.8);

        \draw[blackarrow] (-7.5,0) -- node[above,xshift=-0.2cm] {\footnotesize{$(1-p) \Lambda$}} (-6.5,0);
        \draw[blackarrow] (-7.5,3) -- node[above,xshift=-0.2cm] {\footnotesize{$p \Lambda$}} (-6.5,3);

        \draw [greenarrow] (S1e) -- node[left] {\footnotesize{\textcolor{black}{$ \psi$}}}(P);

        \draw[orangearrow] (-5.5, 2.8) -- node[left,xshift=-0.1cm] {\footnotesize{\textcolor{black}{$ q\omega$}}}(-3.1,0.5);

        \draw[orangefitarrow] (-6.5, 2.8) -- (-8,2.8);
        \draw[orangefitarrow] (-8,2.8) -- (-8,-3.2);
        \draw[orangearrow] (-8,-3.2) -- node[below,xshift=-0.1cm] {\footnotesize{\textcolor{black}{$ (1-q)\omega$}}} (-3.5,-3.2);
        

            \end{tikzpicture}
            \caption{Flowchart of the model~\eqref{eq: sys-am}. Blue arrows represent transmission, recovery, and the loss of natural immunity. Black arrows indicate transitions between age groups, recruitments and deaths.  Green arrows represent administration of Nirsevimab to infants, while orange arrows denote the waning of protection provided by Nirsevimab.}
            \label{fig: flowchart}
        \end{figure}

\subsection{Transmission dynamics}
Viral transmission occurs through contact between susceptible and infected individuals across all age groups. The transmission is described by a mass-action incidence term, i.e., the force of infection acting on the $i$-th group, $\text{FoI}_i$, is given by:
\begin{equation}\label{eq: FoI}
    \text{FoI}_i = \sum_{j=1}^3 \beta_{ij}I_j, \quad i =1,2,3,
\end{equation}
where the age-specific transmission rate $\beta_{ij} = c_{ij} \pi_{ij}$ is defined as the product of the contact rate $c_{ij}$ between individuals in age groups $i$ and $j$, and the probability $\pi_{ij}$ of successful transmission given contact with an infectious individual from group $j$.

The use of a mass-action incidence, also known as density-dependent incidence, assumes that individual contact rates increase linearly with population density. This modelling assumption is well-suited for describing the transmission dynamics of diseases spread through droplets ~\cite{martcheva2015introduction,diekmann2000mathematical}, as is the case for RSV transmission~\cite{kutter2018transmission,MemorialSloanInfectious}. Therefore, this is a widely accepted assumption for RSV models~\cite{weber2001modeling,arenas2008existence,acedo2010mathematical,moore2014modelling,hogan2016exploring}, although standard incidence (also referred to as frequency-dependent incidence) is also sometimes used~\cite{greenhalgh2009backward,hogan2017potential,hodgson2020evaluating}. However, the latter is mostly used when the number of contacts per infectious individual is not expected to increase proportionally with population density (for example, sexually transmitted diseases)~\cite{diekmann2000mathematical}.

The contact patterns are described by a contact matrix $ C = (c_{ij})_{i,j=1}^{3}$, where each element $c_{ij}$ denotes the average number of contacts per unit of time that an individual in age group $j$ has with individuals in age group $i$. In our framework, we neglect the contact rates $c_{13}$ and $c_{23}$, i.e., contacts from seniors to infants and children--and--adults, respectively. This assumption reflects the limited number of contacts that senior adults have with younger adults and infants. It is based on empirical evidence: the primary settings in which contacts take place are families, schools and workplaces. However, individuals with 65 years or more are not present in any school or work settings, so their contacts with adults or children are only in the household setting. However, empirical studies \cite{prem2017projecting,mossong2008social} show that in industrialised countries, the family structure is predominantly two-generational, so the rates $c_{13}$ and $c_{23}$ are negligible compared to the others.

After contracting the virus, infected individuals recover and move from compartment $I_i$ to compartment $R_i$ at rate $\gamma_i$, which corresponds to the inverse of the average infectious period. Once in compartment $R_i$, individuals are assumed to acquire natural immunity to the disease. However, this immunity wanes over time. In age groups 2 and 3, recovered individuals return to compartment $S_i$ at rate $\nu_i$, which is the inverse of the average duration of natural immunity, whereas in the first age group they return to compartment $S_{1w}$ at rate $\nu_1$.

\subsection{The balance equations}
To formalise the epidemiological processes described above, we express the model as a system of balance equations that describe how the population evolves over time. Each equation represents the temporal change in the number of individuals in a given epidemiological compartment and incorporates all mechanisms through which individuals enter or leave that compartment, including infection, recovery, loss of immunity, immunisation, ageing, births, and deaths.

All state variables are functions of time and denote the size of the corresponding population group at any given moment. For each compartment, the corresponding equation is constructed by accounting for the contributions of the relevant biological and demographic processes, so that the rate of change reflects the net effect of all inflows and outflows.

This formulation results in a system of ordinary differential equations, where the time derivative of each state variable represents its instantaneous rate of variation. Collectively, these equations provide a complete and transparent description of the RSV transmission dynamics under the assumptions introduced above:
\begin{subequations}
\allowdisplaybreaks
	\label{eq: sys-am} 
	\begin{align}
		\label{eq: sys-am1}
        \dot{P} &= p\Lambda + \psi S_{1e} -\sigma P(\beta_{11}I_1 + \beta_{12}I_2 )  -(\mu_1+\omega)P,\\
    \label{eq: sys-am2}
    \dot{S}_{1e} &=(1-p)\Lambda - S_{1e}(\beta_{11}I_1 + \beta_{12}I_2  )  -(\mu_1+\eta_1+\psi)S_{1e},\\
    \label{eq: sys-am3}
    \dot{S}_{1w} &= q\omega P- S_{1w}(\beta_{11}I_1 + \beta_{12}I_2 ) - (\mu_1+\eta_1) S_{1w} + \nu_1 R_1,\\
    \label{eq: sys-am4}
    \dot{I}_1 &= \left(\sigma P
    + S_{1e} + S_{1w}\right)(\beta_{11}I_1 + \beta_{12}I_2 )- (\mu_1+\gamma_1+d_1+\eta_1)I_1,\\
    \label{eq: sys-am5}
    \dot{R}_1 &= \gamma_1 I_1  - (\mu_1+\eta_1+\nu_1)R_1,\\
    \label{eq: sys-am6}
    \dot{S}_2 &=  (1-q)\omega P + \eta_1(S_{1e}+S_{1w}) - S_2(\beta_{21}I_1 + \beta_{22}I_2 ) - (\mu_2+\eta_2)S_2+ \nu_2 R_2,\\
    \label{eq: sys-am7}
    \dot{I}_2 &= \eta_1I_1 + S_2(\beta_{21}I_1 + \beta_{22}I_2 ) - (\mu_2+\gamma_2+\eta_2)I_2, \\
    \label{eq: sys-am8}
    \dot{R}_2 &= \eta_1R_1 + \gamma_2 I_2 - (\mu_2+\eta_2+\nu_2)R_2,\\
    \label{eq: sys-am9}
    \dot{S}_3 &= \eta_2S_2 -  S_3( \beta_{31}I_1+\beta_{32}I_2 + \beta_{33}I_3) - \mu_3 S_3+\nu_3 R_3, \\
    \label{eq: sys-am10}
    \dot{I}_3 &= \eta_2I_2+ S_3( \beta_{31}I_1 +\beta_{32}I_2 + \beta_{33}I_3) - (\mu_3+\gamma_3+d_3)I_3,\\
    \label{eq: sys-am11}
    \dot{R}_3 &=  \eta_2R_2 +\gamma_3 I_3 - (\mu_3+\nu_3)R_3,
	\end{align}
\end{subequations}

\noindent where the upper dots denote the time derivatives. The system is equipped with the following initial conditions:
\begin{equation} \label{eq: IC-am}
\begin{gathered}
  P(0)\geq 0, \quad S_{1e}(0) \geq 0, \quad S_{1w}(0) \geq 0, \quad I_1(0) \geq 0, \quad R_1(0) \geq 0, \\
  S_i(0) \geq 0, \quad I_i(0) \geq 0, \quad R_i(0) \geq 0, \quad i =2,3.
  \end{gathered}
\end{equation}

Finally, note that system \eqref{eq: sys-am} -- \eqref{eq: IC-am} can be rewritten in the vector form:
\begin{equation*}
	\dot{\mathbf{x}}(t) = \mathbf{F}\bigl(\mathbf{x}(t)\bigr), \qquad \mathbf{x}(0) = \mathbf{x}_0,
\end{equation*}
where $\mathbf{x}(t) \in \mathbb{R}_+^{11}$ is the vector of the state variables and $\mathbf{F}(\mathbf{x}) = (f_i(\mathbf{x}(t)))_{i=1, \dots, 11}$ denotes the autonomous vector field.

\begin{table}
    \centering
    \footnotesize
    \renewcommand{\arraystretch}{1.2}
    \begin{tabular}{
        >{\centering\arraybackslash}p{1.1cm} 
        >{\raggedright\arraybackslash}p{7.9cm} 
        >{\centering\arraybackslash}p{4.5cm} 
    }
        \toprule
        \textbf{Par.} & \textbf{Description} & \textbf{Baseline value} \\
        \midrule
        $t_f$ & Temporal horizon & $5$ years \\
        
        $\Lambda$ & Recruitment of newborns & $1{,}013.48$ $\text{ind} \cdot \text{days}^{-1}$  \\

        $p$ & Fraction of infants who receive Nirsevimab at birth & $0.475$ \\ 

        $q$ & Fraction of infants losing Nirsevimab protection before 1 year & $0.08$ \\

        $\omega$ & Rate of loss of immunity given by Nirsevimab & $0.56 \cdot 10^{-2}$ days$^{-1}$ \\ 

        $\sigma$ & Factor of ineffectiveness of Nirsevimab &  0.25  \\

        $\psi$ & Rate of administration of Nirsevimab after birth in infants & $0.44 \cdot 10^{-2}$ days$^{-1}$  \\

        $\mu_1$ & Natural death rate of infants & $7.05 \cdot 10^{-6}$ days$^{-1}$  \\

        $\mu_2$ & Natural death rate of children--and--adults & $5.48 \cdot 10^{-6}$ days$^{-1}$ \\

        $\mu_3$ & Natural death rate of seniors & $1.18 \cdot 10^{-4}$ days$^{-1}$ \\

        $d_1$ & Disease-induced death rate of infants & $2.75 \cdot 10^{-5}$ days$^{-1}$ \\

        $d_3$ & Disease-induced death rate of seniors & $2.2 \cdot 10^{-5}$ days$^{-1}$ \\

        $\eta_1$ & Ageing rate from infants to children--and--adults & $0.27 \cdot 10^{-2}$ days$^{-1}$ \\

        $\eta_2$ & Ageing rate from children--and--adults to seniors & $4.28 \cdot 10^{-5}$ days$^{-1}$ \\

        $\beta_{11}$ & Transmission rate from infants to infants & $1.28\cdot10^{-10}$ $\text{ind}^{-1} \cdot \text{days}^{-1}$ \\

        $\beta_{12}$ & Transmission rate from children--and--adults to infants & $3.52 \cdot 10^{-8}$ $\text{ind}^{-1} \cdot \text{days}^{-1}$  \\

        $\beta_{21}$ & Transmission rate from infants to children--and--adults & $5.05 \cdot 10^{-9}$ $\text{ind}^{-1} \cdot \text{days}^{-1}$ \\

        $\beta_{22}$ & Transmission rate from children--and--adults to children--and--adults & $9.3 \cdot 10^{-9}$ $\text{ind}^{-1} \cdot \text{days}^{-1}$ \\

        $\beta_{31}$ & Transmission rate from infants to seniors & $5 \cdot 10^{-10}$ $\text{ind}^{-1} \cdot \text{days}^{-1}$ \\

        $\beta_{32}$ & Transmission rate from children--and--adults to seniors & $8 \cdot 10^{-9}$ $\text{ind}^{-1} \cdot \text{days}^{-1}$ \\

        $\beta_{33}$ & Transmission rate from seniors to seniors & $10^{-9}$ $\text{ind}^{-1} \cdot \text{days}^{-1}$  \\

        $\gamma_1$ & Recovery rate of infants & $0.667 \cdot 10^{-1}$ days$^{-1}$  \\

        $\gamma_2$ & Recovery rate of children--and--adults & $0.2$ days$^{-1}$ \\

        $\gamma_3$ & Recovery rate of seniors & $0.125$ days$^{-1}$ \\

        $\nu_1$ & Immunity loss rate of infants & $0.43 \cdot 10^{-2}$ days$^{-1}$ \\

        $\nu_2$ & Immunity loss rate of children--and--adults & $0.43 \cdot 10^{-2}$ days$^{-1}$ \\

        $\nu_3$ & Immunity loss rate of seniors & $0.43 \cdot 10^{-2}$ days$^{-1}$  \\

        \bottomrule
    \end{tabular}
    \caption{Synoptic summary of parameters in system~\eqref{eq: sys-am}, with description and baseline values. Demographic parameter values (i.e., $\Lambda$ and $\mu_i$, $i=1,2,3$) are obtained from the Italian National Institute of Statistics (ISTAT)~\cite{demoincifre}. Epidemiological parameters values for $\gamma_i$ and $\nu_i$, $i=1,2,3$ are obtained from \cite{li2025relating} and \cite{acedo2010mathematical}, respectively. All the other values are assumed or derived as explained in Section \ref{sec:param}.}
    \label{tab:par}
\end{table}

\section{Qualitative analysis} \label{sec:rep_num}

As a preliminary step, we establish the well-posedness of system~\eqref{eq: sys-am} and identify a bounded feasible region containing all biologically meaningful trajectories. This shows that the model generates a dynamical system on a positively invariant and attracting subset of $\mathbb{R}_+^{11}$, so that the qualitative analysis can be carried out by restricting to that region.

\begin{prop} \label{prop:invariance} The region
\begin{equation*}
\Omega=\left\{\mathbf{x}\in \mathbb{R}_+^{11}:\sum_{i=1}^{11}x_i\leq \frac{\Lambda}{\min\{\mu_1,\mu_2,\mu_3\}}\right\}
\end{equation*}
is positively invariant and globally attracting for system \eqref{eq: sys-am}--\eqref{eq: IC-am}.
\end{prop}

\begin{proof}
See Appendix A.
\end{proof}

System~\eqref{eq: sys-am} admits a disease-free equilibrium, denoted by
\begin{equation} \label{eq:DFE-am}
    E_0 = \bigl(P^0, S_{1e}^0, S_{1w}^0, 0, 0, S_2^0, 0, 0, S_3^0, 0, 0 \bigr),
\end{equation}
where the coordinates are given explicitly by
\begin{equation} \label{eq: coord DFE}
\begin{aligned}
    P^0 &= \frac{p \Lambda + \psi S_{1e}^0}{\omega + \mu_1}, \quad
    S_{1e}^0 = \frac{(1-p) \Lambda}{\mu_1 + \eta_1 +\psi}, \quad S_{1w}^0=\frac{q\omega P^0}{\mu_1+\eta_1} \\
    S_2^0 &= \frac{(1-q)\omega P^0+\eta_1(S_{1e}^0+S_{1w}^0)}{\mu_2 + \eta_2}, \quad
    S_3^0 = \frac{\eta_2 S_2^0}{\mu_3}.
\end{aligned}
\end{equation}

To assess the transmissibility of an infectious disease within a population, we introduce the \emph{basic reproduction number}, denoted by $\Rzero$, which represents the mean number of new infections caused by a single infectious individual over its infectious lifetime in a completely susceptible population, in the absence of any control measures or immunity~\cite{diekmann1990definition}.

When intervention strategies are incorporated into the model framework, the relevant threshold quantity is the \emph{control reproduction number}, indicated by $\Rc$~\cite{gumel2004modelling}. By definition, the presence of control measures reduces transmission, and thus $\Rc<\Rzero$.

These quantities play a fundamental role in determining the invasion dynamics of the disease: sustained transmission is possible only when the relevant reproduction number exceeds unity, whereas values below one correspond to disease elimination.

Following the standard next–generation matrix (NGM) approach, the reproduction numbers are derived from the linearisation of the infection subsystem around the disease-free equilibrium (DFE)~\cite{diekmann1990definition,van2002reproduction}. This framework allows the complex transmission processes across age groups to be summarised into a single threshold quantity that captures the average number of secondary infections generated in the population.

In the present model, three infectious compartments are considered, $I_1, I_2$ and $I_3$, corresponding to the three age classes. The $\mathbf{NGM}$ is therefore a $3\times3$ matrix, whose $(i,j)$ entries quantify the expected number of new infections in age class $i$ generated by an infectious individual in age class $j$, evaluated at the DFE. The control reproduction number, $\Rc$, is defined as the spectral radius (i.e., dominant eigenvalue) of this matrix. Formally, the next-generation matrix is defined as $\mathbf{NGM}=\mathbf{FV}^{-1}$, where $\mathbf{F}$ is the \emph{transmission matrix}, collecting the new--infections terms in each infectious compartment, and $\mathbf{V}$ is the \emph{transition matrix}, accounting for all other processes affecting infected individuals, such as recovery, ageing between age classes, and mortality. Both matrices are obtained as Jacobians of the corresponding terms and are evaluated at the disease-free equilibrium. 

As a result, the control reproduction number can be expressed as
\begin{equation*}\label{eq:Rc}
    \Rc= \max \{\Rc^{(1,2)}, \Rc^{(3)} \},
\end{equation*}
where $\Rc^{(1,2)}$ and $\Rc^{(3)}$ denote the contributions to the control reproduction number associated with transmission processes involving the first two and the third age class, respectively. These quantities are given by
\begin{equation*} \label{eq: Rc1 e Rc2}
    \Rc^{(1,2)} =\frac{1}{2} \left[ \frac{A_{11}}{\xi_1} + \frac{A_{12}\eta_1}{\xi_1\xi_2} + \frac{A_{22}}{\xi_2} + \sqrt{\left(\frac{A_{11}}{\xi_1} + \frac{A_{12}\eta_1}{\xi_1\xi_2} - \frac{A_{22}}{\xi_2}\right)^2 + 4\frac{A_{12}}{\xi_2} \left(\frac{A_{21}}{\xi_1} + \frac{A_{22}\eta_1}{\xi_1\xi_2}\right)}\;\right],
\end{equation*}
\begin{equation*}
    \Rc^{(3)} =\frac{\beta_{33}S_3^0}{\xi_3},
\end{equation*}
where, the coefficients $A_{ij}$ describe the effective transmission from age class $j$ to age class $i$ at the DFE and reflect contact rates, transmission parameters, and the distribution of susceptible individuals across age classes:
\begin{equation*} \label{eq: def Aij}
\begin{aligned}
    A_{11}&=\beta_{11}(\sigma P^0+S^0_{1e}+S^0_{1w}), & A_{12}&=\beta_{12}(\sigma P^0+S^0_{1e}+S^0_{1w}), \\ A_{21}&=\beta_{21}S^0_2, &
    A_{22}&=\beta_{22}S^0_2,\\  A_{31}&=\beta_{31}S_3^0, & A_{32}&=\beta_{32}S^0_3, &A_{33}=\beta_{33} S^0_3,
\end{aligned}
\end{equation*}
while $\xi_i$ represent the total rates at which individuals leave the corresponding infectious compartments due to recovery, ageing, natural mortality, or disease-induced mortality:
\begin{equation}   \label{eq: xi_i}
\xi_1=\mu_1+\gamma_1+d_1+\eta_1, \quad \xi_2=\mu_2+\gamma_2+\eta_2, \quad \xi_3=\mu_3+\gamma_3+d_3.
\end{equation}

The basic reproduction number $\Rzero$ is obtained from the same next--generation matrix construction by setting the intervention--related parameter $p=\psi=0$, thus removing infant protection.

In this case, the disease--free equilibrium reduces to
\[
\bar E_0 = \left(0,\bar S_{1e}^0,0,0,0,\bar S_2^0,0,0,\bar S_3^0,0,0\right),
\]
corresponding to a fully susceptible population.

The expression of $\Rzero$ has the same form as $\Rc$, with the coefficients $A_{ij}$ evaluated at $\bar E_0$ instead of $E_0$. The explicit expressions of $\bar S_{1e}^0$, $\bar S_2^0$, and $\bar S_3^0$ are reported in \autoref{sec: DFE}, where a detailed derivation of the next--generation matrix and a biological interpretation of the reproduction numbers are also provided.

The role of the control reproduction number $\Rc$ as a threshold parameter for system~\eqref{eq: sys-am} is characterised by two results corresponding to the cases $\Rc<1$ and $\Rc>1$. 

\begin{prop} \label{thm: LAS DFE}
    Suppose that $\Rc<1$. Then the disease-free equilibrium $E_0$ of system \eqref{eq: sys-am} is locally asymptotically stable.
\end{prop}
\begin{proof}
    See \autoref{sec:LAS_DFE}.
\end{proof}

To analyse the case $\Rc>1$, we first introduce the state space $X$ and the biologically feasible region $X_0$, defined, respectively, as:
\begin{align*}
    X=\{ ( P, S_{1e},& S_{1w}, I_1, R_1,S_2,I_2,R_2,S_3,I_3,R_3)\in\R^{11} \colon \\ &P, S_{1e}, S_{1w}, S_2,S_3, I_i,R_i \geq 0,\, i=1,2,3 \},
\end{align*}
and
\begin{equation*}
    X_0=\{ ( P, S_{1e}, S_{1w}, I_1, R_1,S_2,I_2,R_2,S_3,I_3,R_3) \in X \colon I_i>0,\, i=1,2,3\}.
\end{equation*}

Note that $\Omega \subset X$, where $\Omega$ is the positively invariant and attractive set defined in Proposition~\ref{prop:invariance}. Hence, all trajectories of system~\eqref{eq: sys-am} eventually evolve within $X$.

We, then, define the boundary of the feasible region:
\begin{equation*}
\partial X_0= X \setminus X_0.
\end{equation*}

It is straightforward to verify that both $X$ and $X_0$ are positively invariant sets for system \eqref{eq: sys-am}. This allows us to introduce the following notion of uniform persistence~\cite{thieme1993persistence}.

\begin{definition}
    System \eqref{eq: sys-am} is said to be \emph{uniformly persistent} with respect to the pair $(X_0, \partial X_0)$ if there exists a constant $\delta>0$ such that for any solution of the system with initial conditions in $X_0$ the following holds: 
    \begin{equation*}
        \liminf_{t\to \infty} I_i(t)>\delta \quad i=1,2,3.
    \end{equation*}
\end{definition}

To prove uniform persistence of system~\eqref{eq: sys-am} with respect to the pair $(X_0,\partial X_0)$, we define the invariant boundary set (or \emph{disease--free invariant set})
\begin{align*} 
    M_{\partial} = \{ &(P(0),S_{1e}(0), S_{1w}(0), S_2(0), S_3(0), I_i(0), R_i(0)) \in \partial X_0 \colon\\  
    & \Phi_t (P(0), S_{1e}(0), S_{1w}(0), S_2(0), S_3(0), I_i(0), R_i(0)) \in \partial X_0, \forall t \geq 0,i=1,2,3\}
\end{align*}
where $\Phi_t : X \rightarrow X$ is the semiflow defined by system \eqref{eq: sys-am}. 

Uniform persistence when $\Rc>1$ is established in the following proposition.

\begin{prop} \label{prop:unif_persist} 
    If $\Rc>1$, then system \eqref{eq: sys-am} is uniformly persistent with respect to the pair $(X_0, \partial X_0)$.
\end{prop}
    \begin{proof}
        To prove the Proposition, it is sufficient to verify that the hypotheses of Theorem 1.3.1 in \cite{zhao2003dynamical} are satisfied. Firstly, note that both $X$ and $X_0$ are positively invariant. Moreover, $\Phi_t$ is compact and point-dissipative. These properties guarantee the existence of a global attractor for $\Phi_t$ satisfying condition (C1) in \cite{zhao2003dynamical}. In order to verify the condition (C2), from the Lemma~\ref{Lemma:M=Mpartial} (see \autoref{sec:permanence}) we have that $M_{\partial}$ is the maximal compact invariant set of $\Phi_t$ in $\partial X_0$. Moreover, choosing Morse decomposition of $M_{\partial}$ as $\{E_0\}$, it means that $\{E_0\}$ is isolated. Lemma~\ref{Lemma:E0_weakrep} (see \autoref{sec:permanence}) implies $W^s(E_0) \cap X_0 = \emptyset$, hence the condition (C2) in \cite{zhao2003dynamical} is satisfied. Then, by the Theorem 1.3.1 in \cite{zhao2003dynamical}, with $L=X_0$, system~\eqref{eq: sys-am} is uniformly persistent with respect to $(X_0, \partial X_0)$. 
    \end{proof}

\begin{remark}
Uniform persistence immediately implies the existence of a stationary coexistence state. Indeed, by Theorem 1.3.7 in \cite{zhao2003dynamical}, the semiflow $\Phi_t$ admits at least one equilibrium in $X_0$. Consequently, whenever $\Rc>1$, system~\eqref{eq: sys-am} has at least one endemic equilibrium.
\end{remark}

For completeness, it is worth noting that the model also admits a theoretical possibility for the disease to circulate exclusively within the senior age group. This phenomenon, unlikely to occur in a real-world epidemiological context, is mathematically discussed in \autoref{sec:senior-only}.

\section{Parametrisation}\label{sec:param}

In this section, we describe the parametrisation adopted for model \eqref{eq: sys-am}. For each parameter, we summarise the rationale underlying the chosen baseline value and indicate the corresponding source or modelling assumption. The demographic component of the parametrisation is informed by official Italian data for 2024 \cite{demoincifre}, whereas the epidemiological and intervention-related parameters are taken from the clinical and epidemiological literature or introduced through explicit modelling assumptions. A complete list of parameter values and the simulation time horizon is reported in Table \ref{tab:par}. 

\subsection{Demographic parameters}
The simulations are carried out over a five-year period, allowing us to capture several consecutive epidemic seasons. We employ a daily time step, which yields a final simulation time of $t_f = 5$ years, corresponding to $5\cdot365 = 1{,}825$ days.

The demographic parameter values were all obtained from the Italian National Institute of Statistics (ISTAT) \cite{demoincifre}. According to the 2024 data, a total of $369{,}922$ live births were recorded, 
yielding a daily recruitment rate of $\Lambda = 369{,}922/365 \approx 1{,}013.48$ $\text{ind}\cdot\text{days}^{-1}$.  The annual mortality probabilities are 
$\mu_{1,\mathrm{ann}} = 2.57 \cdot 10^{-3}$ for infants, 
$\mu_{2,\mathrm{ann}} = 0.2 \cdot 10^{-2}$ for children--and--adults, 
and $\mu_{3,\mathrm{ann}} = 4.2 \cdot 10^{-2}$ for seniors. The corresponding daily mortality rates for each age class are computed using the standard probability-to-rate conversion formula:
\[
\mu_i = -\frac{\log\left( 1 - \mu_{i,\mathrm{ann}} \right)}{365}\;\text{days}^{-1}, \quad i=1,2,3.
\]
Thus, $\mu_1 \approx 7.05\cdot10^{-6}$ days$^{-1}$, $\mu_2 \approx 5.48\cdot10^{-6}$ days$^{-1}$, and $\mu_3 \approx 1.18\cdot10^{-4}$ days$^{-1}$.

The mean time spent in the first age class is one year (i.e. 365 days), so the ageing rate from the first to the second class is $\eta_1 = 1/365 = 0.27 \cdot 10^{-2}$ days$^{-1}$. The mean time spent in the second age class is $64$ years, and therefore the transition rate from the second to the third class 
is $\eta_2 = 1/(365 \cdot 64) \approx 4.28 \cdot 10^{-5}$ days$^{-1}$.

\subsection{Protection--related parameters}\label{subsec:prot_par}
The parameter values presented in this section are based on modelling assumptions. 
Let p denote the fraction of all infants born in a year who receive Nirsevimab at birth. In the Italian setting, Nirsevimab is administered at birth only during the epidemic season; hence
p depends both on the timing of births and on the coverage within the epidemic season.  This quantity can be expressed as 
\begin{equation}\label{eq p}
p=\pi c
\end{equation}
where $\pi$ is the fraction of births occurring during the epidemic season and $c$ denotes the fraction of infants receiving Nirsevimab at birth among those born during the epidemic season. Assuming that births are uniformly distributed throughout the year and that the epidemic season lasts six months, we have $\pi = 0.5$, and since $c\leq 1$, it follows that $p\leq 0.5$, which represents the theoretical upper bound for the protection-at-birth coverage. If the coverage is $95\%$, then $c = 0.95$, yielding $p = 0.475$.
Infants born outside the epidemic season are assumed to be eligible for antibody administration during the subsequent epidemic season. We denote by $\psi_{\mathrm{ann}}$ the annual probability that such infants receive Nirsevimab, and assume $\psi_{\mathrm{ann}}=0.8$, corresponding to an 80\% annual coverage among infants born outside the epidemic season. Using the standard probability-to-rate conversion formula, the associated daily protection rate is therefore defined as $\psi = -\log(1 - \psi_{\mathrm{ann}})/365$. For $\psi_{\mathrm{ann}}=0.8$, this yields $\psi = 0.44 \cdot 10^{-2}$ days$^{-1}$.

We also assume that Nirsevimab does not provide full protection and has an efficacy of 75\% \cite{efficacy_nirsevimab,menegale2025impact,shi2026respiratory}, 
so the ineffectiveness factor is 0.25. 
Protection is assumed to last six months (180 days), giving a waning rate of $\omega = 1/180 = 0.56 \cdot 10^{-2}$ days$^{-1}$.

When protection wanes, a fraction $q$ of infants are still under one year of age and consequently enter the susceptible compartment of the first age class, while the remaining fraction $1-q$ moves to the susceptible compartment of the second age class. In our numerical implementation, we set $q = 0.08$.

\subsection{Epidemiological parameters}
According to Li et al.~\cite{li2025relating}, the average duration of infection is 15 days in infants, 5 days in children--and--adults, and 8 days in seniors. Therefore, the recovery rate for each age class is $\gamma_1 = 1/15 = 0.667 \cdot 10^{-1}$ days$^{-1}$, $\gamma_2 = 1/5 = 0.2$ days$^{-1}$, and $\gamma_3 = 1/8 = 0.125$ days$^{-1}$, respectively.

Acedo et al.~\cite{acedo2010mathematical} report that the loss of disease-induced immunity is the same across all age classes and lasts 230 days, so the corresponding waning rate is $\nu_1 = \nu_2 = \nu_3 = 1/230 = 0.43 \cdot 10^{-2}$ days$^{-1}$.

We also account for disease-induced mortality in the first and third age classes. In the first class, we suppose an annual probability of death of $d_{1,\mathrm{ann}} = 0.01$, while the annual probability of death in the third class is 
assumed to be $d_{3,\mathrm{ann}} = 0.008$. Hence, the disease-induced mortality rate in these two age classes is computed as
\[
d_i = -\log\left(1 - d_{i,\mathrm{ann}}\right)/365\;\; \text{days}^{-1}, \qquad i = 1, 3,
\]
which yields to $d_1 \approx 2.75 \cdot 10^{-5}$ days$^{-1}$ and $d_3 \approx 2.2 \cdot 10^{-5}$ days$^{-1}$.

We also selected the transmission rate parameters to reproduce age-specific annual cumulative incidence levels consistent with estimates reported by the European Respiratory Virus Surveillance Summary (ERVISS)~\cite{erviss}; therefore, we set the values of $\beta_{ij}$, $i,j=1,2,3$, as shown in Table \ref{tab:par}.

\subsection{Initial conditions}\label{subsec:initial_conditions}

The initial conditions used in the numerical simulations are chosen close to the disease--free equilibrium (DFE), so as to represent the introduction of RSV into an otherwise disease--free population and to minimise transient dynamics arising from arbitrary starting states. Specifically, we perturb the DFE by introducing a small number of infectious individuals in each age group while preserving the total population size $N_{i}^0$ at the disease--free equilibrium within each group.

We set the initial numbers of infectious individuals to \(I_1(0)=30\), \(I_2(0)=50\), and \(I_3(0)=20\), and we initialise all recovered compartments at zero, \(R_i(0)=0\) for \(i=1,2,3\).

In the first age group, the infectious seed \(I_1(0)\) is redistributed across the non-infectious infant compartments proportionally to their effective susceptibilities at the disease-free equilibrium. To simplify notation, let
\[
\mathcal{S}_{1,\mathrm{eff}}^0:=\sigma P^0+S_{1e}^0+S_{1w}^0.
\]
This choice is consistent with the infection term in the model and preserves the total population size at the disease-free equilibrium within the first age group. Accordingly, we set
\[
P(0)=P^0-\frac{\sigma P^0}{\mathcal{S}_{1,\mathrm{eff}}^0}\,I_1(0),\quad 
S_{1e}(0)=S_{1e}^0-\frac{S_{1e}^0}{\mathcal{S}_{1,\mathrm{eff}}^0}\,I_1(0),
\quad
S_{1w}(0)=S_{1w}^0-\frac{S_{1w}^0}{\mathcal{S}_{1,\mathrm{eff}}^0}\,I_1(0).
\]
For the remaining age groups, the seed is taken from the susceptible class, giving
\[
S_2(0)=S_{2}^0-I_{2}(0),\qquad
S_3(0)=S_{3}^0-I_{3}(0).
\]

\section{Model outcomes}\label{sec:num_res}

In this section we present the numerical results obtained from the model analysis, focusing on the temporal evolution of cumulative incidence and on the impact of different preventive strategies based on Nirsevimab administered to infants. All results are shown over a five-year time horizon. Cumulative incidences by age group are displayed by means of histograms, where each bar represents the cumulative number of new infections recorded within a year, while contour plots are used to illustrate the dependence of selected outcomes on protection-at-birth and catch-up coverage levels. Different colours correspond to the three age classes considered in the model.

Cumulative incidences are computed as time integrals of the corresponding incidence terms over each yearly interval. Specifically, the annual cumulative incidence in the $i$-th age class during year $Y$, denoted by $\text{CI}_i(Y)$, is defined as the time integral of the inflow into the corresponding infected compartment. For the three age classes considered in the model, this yields~\cite{buonomo2020effects}:
\begin{align*}
\text{CI}_1(Y) &= \int_{Y-1}^{Y} \big(\sigma P(t) + S_{1e}(t) + S_{1w}(t)\big)\,\text{FoI}_1(t)\,\mathrm{d}t, \\
\text{CI}_i(Y) &= \int_{Y-1}^{Y} S_i(t)\,\text{FoI}_i(t)\,\mathrm{d}t, \quad i=2,3,
\end{align*}
where $\text{FoI}_i(t)$, $i=1,2,3$, denotes the force of infection acting on age class $i$, defined by~\eqref{eq: FoI}.

\subsection{Exploring the 2024 Italian RSV outbreak}

\begin{figure}[t!]
    \centering
    \begin{subfigure}{\textwidth}
        \centering
        \includegraphics[width=0.75\textwidth]{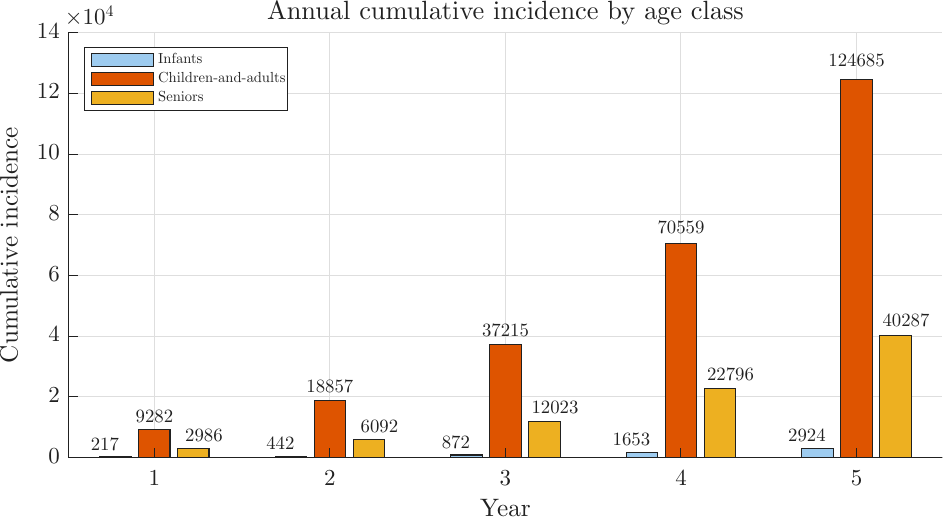}
        \caption{Incidence}
        \label{fig:baseline_inc}
    \end{subfigure}

    \vspace{0.3cm}

    \begin{subfigure}{\textwidth}
        \centering
        \includegraphics[width=0.75\textwidth]{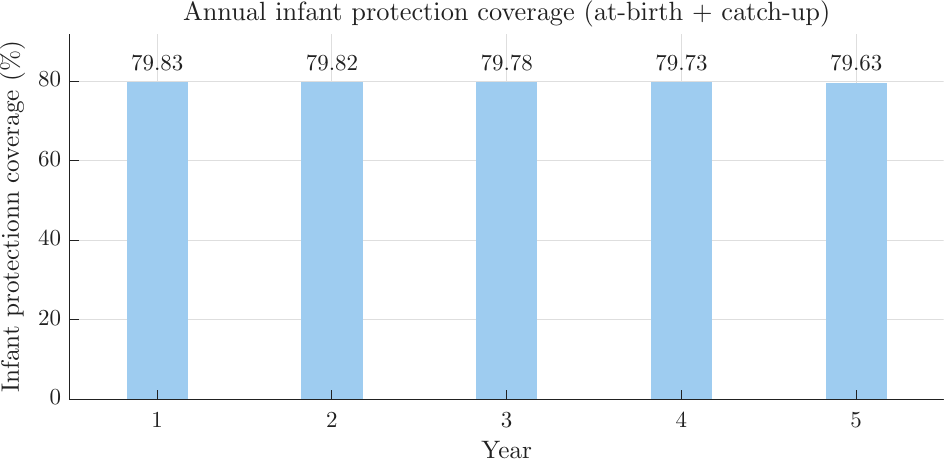}
        \caption{Protection through Nirsevimab coverage}
        \label{fig:baseline_cov}
    \end{subfigure}

     \caption{(a) Annual cumulative number of new RSV cases over a five-year horizon, stratified by age group. The blue bars report infants (cumulative incidence of \(I_1(t)\)), the orange bars report children--and--adults (cumulative incidence of \(I_2(t)\)), and the yellow bars report seniors (cumulative incidence of \(I_3(t)\)). (b) Annual infant coverage (percentage of infants receiving Nirsevimab) over a five-year horizon. Parameter values correspond to the Italian 2024 baseline scenario reported in Table~\ref{tab:par}.}
\end{figure}

We first analyse the model outcomes obtained using parameter values reported in \autoref{tab:par}, representative of the Italian epidemiological situation in 2024–2025, which we refer to as the baseline scenario. Figure \ref{fig:baseline_inc} reports the annual cumulative incidence for each age class over five consecutive years. A clear increasing trend is observed, with the largest contribution arising from the adults population, while infants account for a smaller but epidemiologically relevant fraction of total cases. The magnitude of infant infections is consistent with epidemiological estimates reported by Dovizio et al.~\cite{dovizio2024clinical}. Moreover, these estimates are coherent with surveillance data from the European Respiratory Virus Surveillance Summary (ERVISS)~\cite{erviss}. The progressive growth of cumulative cases reflects the persistence of viral circulation in the population.

Still within the baseline scenario, we report the temporal evolution of infant antibody coverage. We define the annual number of infants receiving Nirsevimab during year $Y$, denoted by $\text{AC}(Y)$, as the time integral of the inflow into the protected compartment $P$. In the model, this inflow is given by the sum of newborns protected at birth, $p\Lambda$, and the rate of subsequent protection of susceptible infants, $\psi S_{1e}$. Therefore,
\begin{equation*}
    \text{AC}(Y) = \int_{Y-1}^{Y} \big(p\Lambda + \psi S_{1e}(t) \big)\,\mathrm{d}t.
\end{equation*}

The histogram in Figure~\ref{fig:baseline_cov} shows a stable pattern over the five-year horizon. This behaviour reflects the near-stationary balance between demographic turnover and immunisation processes in the infant population under the baseline assumptions, which leads to minimal variation over the five-year period.

\subsection{Impact of infant protection coverages on RSV transmission}
\label{5.2}
\begin{figure}[t!]
    \centering
	\subfloat[]{
	\includegraphics[width=0.48\linewidth]{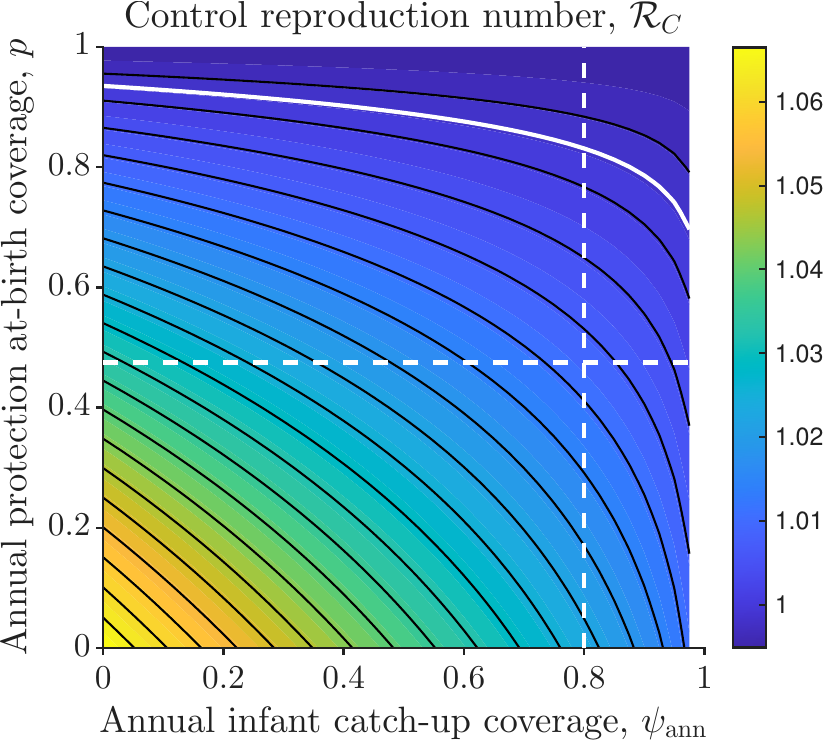}
    \label{fig:contour_Rc}
	}
	\subfloat[]{	
    \includegraphics[width=0.48\linewidth]{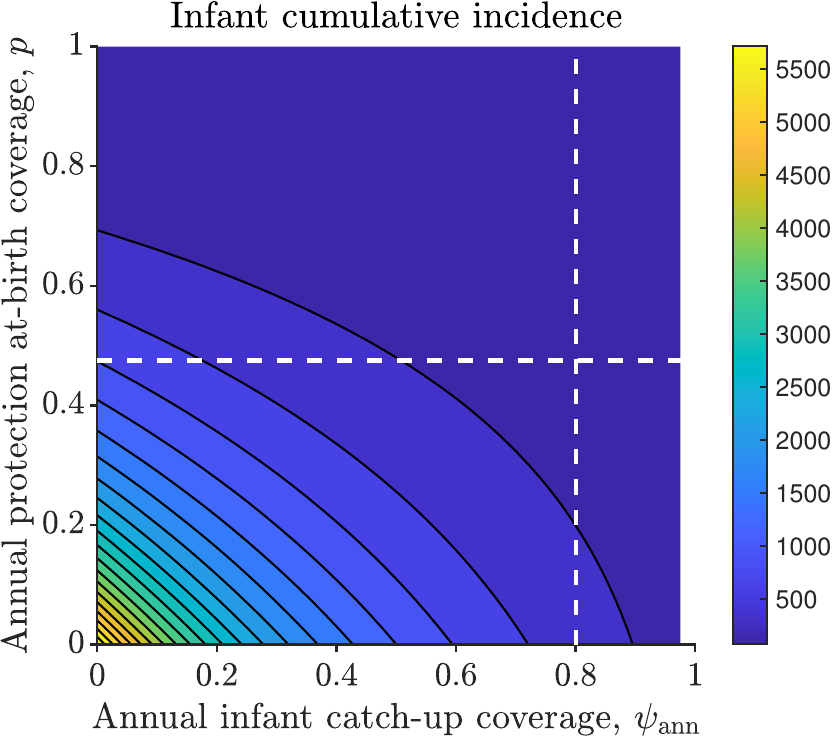}
    \label{fig:contour_CI}
    }
    \caption{(a) Contours of the control reproduction number $R_C$ as a function of infant annual protection at birth coverage, $p$, and infant annual catch-up coverage, $\psi_{\mathrm{ann}}$. (b) Contours of infant cumulative incidence during the first year of the five-year time horizon, as a function of $p$ and $\psi_{\mathrm{ann}}$. The intersection between the dotted white lines indicates the baseline values, namely $p=0.475$, $\psi_{\mathrm{ann}}=0.8$. Note that, although
    p could theoretically reach 1, the assumptions of Section~\ref{subsec:prot_par}, together with formula \eqref{eq p}, imply $p \leq 0.5$. All the other parameter values are set as in Table~\ref{tab:par}.}
    \label{fig:contour_plots}
\end{figure}

To quantify the epidemiological impact of infant protection strategies, we analyse how varying coverage at birth, $p$, and catch-up coverage, $\psi_{\mathrm{ann}}$, influences the control reproduction number $\Rc$ and infants cumulative incidence. All other parameters are held constant at their baseline values for Italy in 2024, as reported in \autoref{tab:par}. We recall that $\psi_{\mathrm{ann}}$ refers to the annual infant catch-up coverage (see Section~\ref{subsec:prot_par}).

The dependence of $\Rc$ on infant protection coverage is illustrated through the contour plot in Figure~\ref{fig:contour_Rc}. It indicates that there is a critical threshold at $ \Rc = 1 $; however, for the Italian scenario considered, this threshold cannot be reached since $ p < 0.5 $. Therefore, although variations in coverage affect the value of $\Rc$, the reduction remains limited and does not allow the critical threshold to be crossed.

A different picture emerges when analysing the cumulative incidence among infants, shown in Figure~\ref{fig:contour_CI}, where incidence is reported as a function of both protection at-birth coverage, $p$, and infant catch-up coverage, $\psi_{\mathrm{ann}}$. Increasing either component of protection leads to a substantial reduction in the annual number of cases. For low coverage values, infant infections are of the order of 5{,}000 cases per year, while higher combined coverage levels reduce the cumulative incidence to approximately 200 cases. The contour plot highlights that, for low values of both $p$ and $\psi_{\mathrm{ann}}$, cumulative infant infections decrease more steeply and approximately linearly as either parameter increases, with comparable contributions from the two strategies. At higher coverage levels, the relationship becomes non-linear, and increases in one parameter progressively reduce the additional impact of the other on cumulative incidence.

\subsection{Comparing different preventive scenarios}

\begin{figure}
    \centering
    \includegraphics[width=0.75\textwidth]{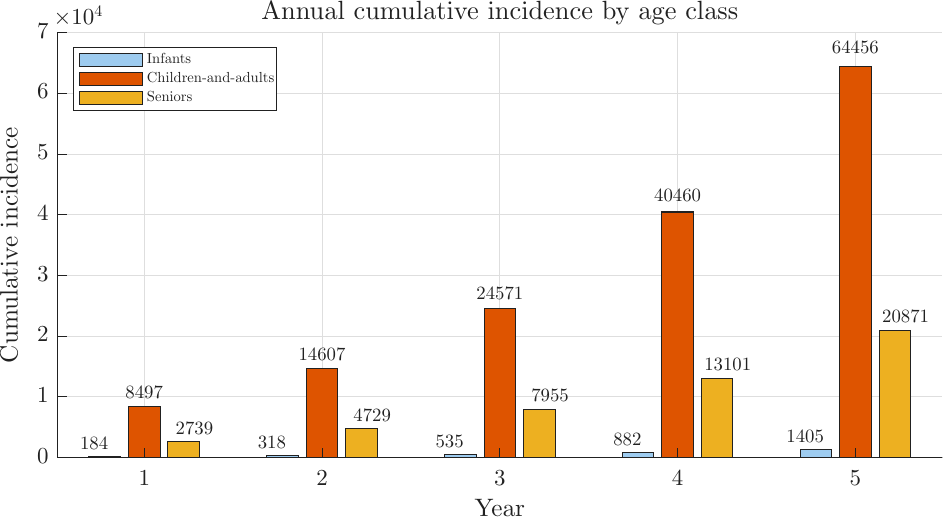}
    \caption{Annual cumulative RSV infections over a five-year horizon under the \emph{upward scenario} (UWS), in which annual infant catch-up coverage, $\psi_{\mathrm{ann}}$, is increased by 10\% relative to the baseline scenario. Each bar corresponds to one year and reports cumulative incidence by age group: infants (blue, $I_1(t)$), children--and--adults (orange, $I_2(t)$), and seniors (yellow, $I_3(t)$). All other parameters remain as in the Italian 2024 baseline scenario (see \autoref{tab:par}).}
    \label{fig:incr_cover}
\end{figure}

\begin{figure}[t!]
    \centering
\includegraphics[width=0.75\textwidth]{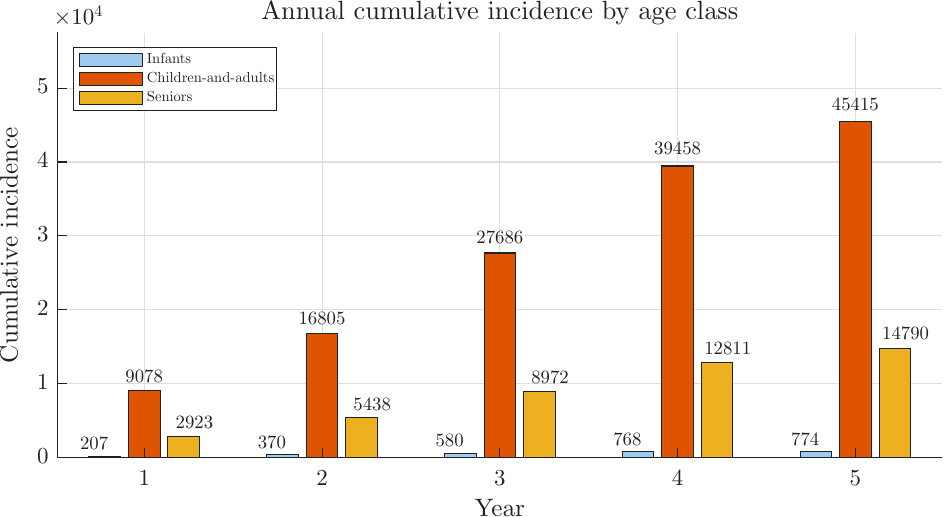}
    \caption{Annual cumulative RSV infections over a five-year horizon under the \emph{annual increase scenario} (AIS), in which annual infant catch-up coverage, $\psi_{\mathrm{ann}}$, is increased by 5\% in the first year relative to baseline and subsequently by 5\% each year relative to the previous year. Each bar corresponds to one year and reports cumulative incidence by age group: infants (blue, $I_1(t)$), children--and--adults (orange, $I_2(t)$), and seniors (yellow, $I_3(t)$). All other parameters remain as in the Italian 2024 baseline scenario (see \autoref{tab:par}).}
    \label{fig:progressive_incr}
\end{figure}

Here, we analyse alternative preventive scenarios and evaluate their impact on cumulative incidence across age groups.

As a first case, we examine the \emph{upward scenario} (UWS), in which infant catch-up coverage is increased by 10\% over the baseline value, while all other model parameters remain unchanged. The corresponding histograms in Figure~\ref{fig:incr_cover} show that, although the overall trend remains increasing over time, the total number of cases is significantly reduced. In particular, after five years, infant cases decrease from 2{,}924 in the baseline scenario to 1{,}405 under the upward scenario. A reduction is also observed in the other age classes, reflecting indirect protection arising from fewer infectious infants who can transmit the virus to seniors.

We then consider a progressive increase in protection coverage for infants born outside the epidemic season, with annual infant catch-up coverage, $\psi_{\mathrm{ann}}$, increasing by 5\% in the first year relative to baseline and subsequently by 5\% each year relative to the previous year. Formally, the annual catch-up coverage evolves according to
\begin{equation*}
    \psi_{\mathrm{ann}}(Y) = 1.05\; \psi_{\mathrm{ann}}(Y-1), \quad Y=1,\dots,5,
\end{equation*}
where $Y$ denotes the simulation year and $\psi_{\mathrm{ann}}(0) = 0.80$ represents the baseline value. This case, which we refer to as the \emph{annual increase scenario} (AIS), is shown in Figure~\ref{fig:progressive_incr}. A substantial reduction in infant cases is observed with respect to the baseline scenario, with corresponding decreases in the other age classes over the five-year horizon, consistently with the indirect protection mechanism described for UWS.

A comparison between the UWS and AIS scenarios reveals a notable temporal pattern. During the first three years, UWS yields fewer cases than AIS across all age groups, reflecting the fact that catch-up coverage, $\psi_{\mathrm{ann}}$, under AIS starts lower than the constant value under UWS. 
From the fourth year onward, the gradually increasing $\psi_{\mathrm{ann}}$ under AIS exceeds the constant UWS value, yet the number of cases remains initially higher in all age groups. This highlights a delayed effect in the AIS scenario, as the improvement relative to UWS emerges consistently from year four onward in every age group.

To further quantify the differences between UWS and AIS, we compute the year-by-year relative variation in cumulative incidence for both UWS and AIS with respect to the baseline scenario, by defining the index
\begin{equation}\label{eq:RX}
    \mathrm{R}X = \frac{X - X_0}{X_0} \cdot 100\%.
\end{equation}

This quantity measures the percentage relative change of $X$, for $X \in \{CI_i(Y) \colon i=1,2,3,\, Y=1,\dots,5\}$, compared with the corresponding value $X_0$ given by model~\eqref{eq: sys-am} in the Italian 2024 baseline scenario. This index provides a numerical quantification of the patterns already observed in the cumulative incidence histograms, allowing a clear comparison of the temporal effects of each scenario across age groups. The results, summarised in \autoref{tab:rx_all}, confirm the temporal patterns observed in the cumulative incidence histograms, with UWS producing a larger reduction than AIS during the first three years, and AIS showing a delayed effect that becomes apparent in the fourth and fifth years.

\begin{table}
\centering
\caption{Year-by-year relative variation in cumulative incidence (CI) for each age group compared with the Italian 2024 baseline scenario, as quantified by the index $\mathrm{R}X$, defined by Eq.~\eqref{eq:RX}. Results are shown for the baseline scenario ($\psi_{\mathrm{ann}}=0.80$) and for two alternative settings. In the \emph{upward scenario} (UWS), $\psi_{\mathrm{ann}}$ is increased by 10\% relative to baseline (i.e., $\psi_{\mathrm{ann}}=0.88$). In the \emph{annual increase scenario} (AIS), $\psi_{\mathrm{ann}}(Y)=1.05\,\psi_{\mathrm{ann}}(Y-1)$ for $Y=1,\dots,5$, where $Y$ denotes the simulation year and $\psi_{\mathrm{ann}}(0)=0.80$ denotes the baseline value. }
\label{tab:rx_all}
\setlength{\tabcolsep}{6pt}
\renewcommand{\arraystretch}{1.05}
\begin{tabular}{lrrrrr}
\toprule
\textbf{Infants} & {$\text{RCI}_1(1)$} & {$\text{RCI}_1(2)$} & {$\text{RCI}_1(3)$} & {$\text{RCI}_1(4)$} & {$\text{RCI}_1(5)$} \\
\midrule
Baseline scenario, $\psi_{\text{ann}}=0.80$ & 0 & 0 & 0 & 0 & 0 \\
Upward scenario, $\psi_{\text{ann}}=0.88$ & -15.2 & -28.1 & -38.6 & -46.6 & -52.0 \\
Annual increase scenario  & -4.6 & -16.3 & -33.5 & -53.5 & -73.5 \\
\midrule
\midrule

\textbf{Children-and-Adults} & {$\text{RCI}_2(1)$} & {$\text{RCI}_2(2)$} & {$\text{RCI}_2(3)$} & {$\text{RCI}_2(4)$} & {$\text{RCI}_2(5)$} \\
\midrule
Baseline scenario, $\psi_{\mathrm{ann}}=0.8$ & 0 & 0 & 0 & 0 & 0 \\
Upward scenario, $\psi_{\mathrm{ann}}=0.88$ & -8.5 & -22.5 & -34.0 & -42.7 & -48.3 \\
Annual increase scenario & -2.2 & -10.9 & -25.6 & -44.1 & -63.6 \\
\midrule
\midrule

\textbf{Seniors} & {$\text{RCI}_3(1)$} & {$\text{RCI}_3(2)$} & {$\text{RCI}_3(3)$} & {$\text{RCI}_3(4)$} & {$\text{RCI}_3(5)$} \\
\midrule
Baseline scenario, $\psi_{\mathrm{ann}}=0.8$ & 0 & 0 & 0 & 0 & 0 \\
Upward scenario, $\psi_{\mathrm{ann}}=0.88$ & -8.2 & -22.4 & -33.8 & -42.5 & -48.2 \\
Annual increase scenario & -2.1 & -10.7 & -25.3 & -43.8 & -63.3 \\
\bottomrule
\end{tabular}
\end{table}

Finally, we consider a scenario in which both protection at birth, $p$, and catch-up coverage, $\psi_{\mathrm{ann}}$, are reduced by 5\% relative to their baseline values. The resulting histograms in Figure~\ref{fig:decr_cover} show a systematic increase in cumulative incidence across all age groups. In infants, the relative increase ranges from approximately +7\% in the first year to +36\% in the fifth year. The other two age classes exhibit a similar temporal pattern, with an increase of about +4\% in the first year that rises to +32\% by the fifth year. These results further emphasise the role of infant protection in shaping transmission dynamics: even moderate reductions in coverage translate into progressively larger increases in incidence over time, not only among infants but also in the broader population.

\begin{figure}[t!]
    \centering
    \includegraphics[width=0.75\textwidth]{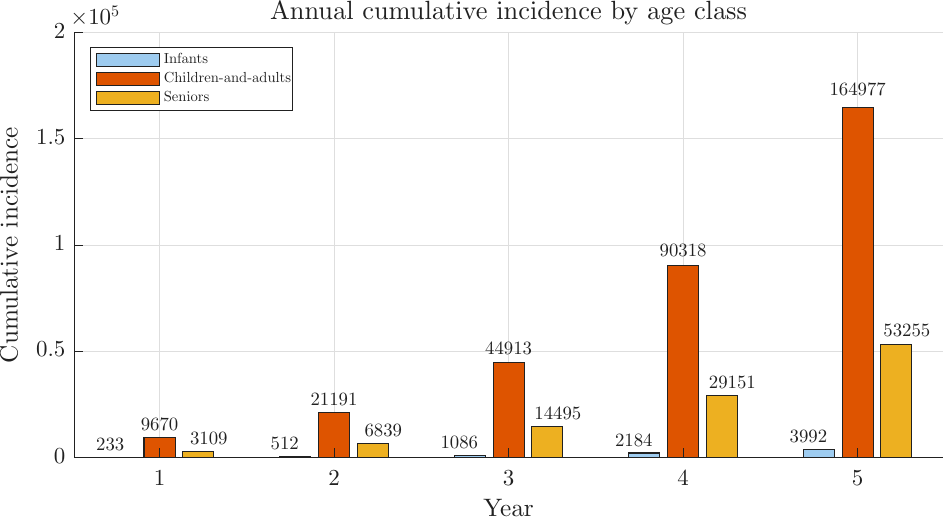}
    \caption{Annual cumulative RSV infections over a five-year horizon under a scenario in which both annual infant protection at birth coverage, $p$, and annual catch-up coverage, $\psi_{\mathrm{ann}}$, are decreased by 5\% relative to baseline. Each bar corresponds to one year and reports cumulative incidence by age group: infants (blue, $I_1(t)$), children--and--adults (orange, $I_2(t)$), and seniors (yellow, $I_3(t)$). All other parameters remain as in the Italian 2024 baseline scenario (see \autoref{tab:par}).}
    \label{fig:decr_cover}
\end{figure}

\section{Conclusions and discussion}\label{sec:conclusion}
In this work, we propose a stage-structured mathematical model for RSV transmission, stratified by age classes and including infant protection via Nirsevimab, a monoclonal antibody administered at birth to infants born during the epidemic season and, through a catch-up strategy in the subsequent epidemic season, to those born outside it, in line with the approach adopted in the Italian setting. We derive analytical expressions for the basic and control reproduction numbers and establish key qualitative properties of the model. We then use numerical analyses to investigate the temporal evolution of RSV incidence and to evaluate the epidemiological impact of alternative infant protection strategies under parameter values representative of the Italian epidemiological context.

The numerical results provide insight into the effects of infant protection strategies, as summarised below:
\begin{itemize}
    \item Infant protection substantially reduces RSV incidence among infants, with marked decreases in cumulative cases observed under increased coverage levels, supporting the effectiveness of Nirsevimab in reducing disease burden in this age class;
    \item As highlighted in Section \ref{5.2}, although infant-targeted protection alone does not reduce $\Rc$ below the epidemic threshold it effectively lowers RSV incidence among infants and provides indirect benefits across age groups;
    \item Scenario analyses show that both the \emph{upward scenario} and the \emph{annual increase scenario} effectively reduce RSV incidence across all age groups. The upward scenario provides immediate reductions, whereas the annual increase scenario produces slightly smaller early effects but greater reductions in later years, as coverage continues to rise. This delayed improvement highlights the potential of progressive strategies to maximize long-term epidemiological benefits.
\end{itemize}
It is worth noting that the European landscape of RSV prevention is heterogeneous and continuously evolving, reflecting differences in the timing of Nirsevimab introduction and in country-specific implementation strategies. Emerging evidence suggesting extended efficacy of Nirsevimab up to at least six months may further influence this scenario \cite{munro2025180}. Furthermore, the parallel introduction of maternal vaccination--already incorporated into national prevention strategies in some European countries--and, in a minority of settings, the adoption of combined preventive approaches, are likely to further impact RSV incidence, transmission dynamics, and clinical burden in childhood. This variability in programmatic approaches, drug availability, and reimbursement policies complicates cross-country comparisons of impact and effectiveness, while at the same time providing a valuable natural framework to assess real-world implementation and the relative effectiveness of different strategies.

An important observation in this regard concerns the structure of current immunisation programmes: all European countries implementing Nirsevimab prophylaxis adopt a seasonal approach \cite{franck2026evolving}, yet not all include a catch-up strategy for infants born outside the epidemic season. Even among those that do, eligibility criteria vary substantially, with upper age limits ranging from 2 to 12 months (and up to 24 months in selected high-risk conditions). In Italy, where an upper age limit of 9 months is generally applied, regional variability has allowed extensions up to 12 months in regions such as Lombardy, Sicily, Piedmont, and Veneto.

In this context, the results of the proposed framework offer a relevant interpretive key. According to the analysis carried out, cumulative RSV incidence is primarily associated with the level of annual infant catch-up coverage, provided that prevention strategies achieve at least 50--60$\%$ coverage of the susceptible population at birth. This finding may help explain the marked variability observed in Italy during the first season of Nirsevimab implementation, where immunization at birth was generally appropriately high (often exceeding 90 \%), while the capacity to reach infants born outside the epidemic season varied substantially across regions, resulting in significant heterogeneity in overall coverage and, consequently, in preventive effectiveness. Moreover, the proposed framework suggests that the lowest RSV incidence is achieved when at-birth coverage reaches approximately 70 \% of the annual birth cohort. This points to the potential of a continuous immunization strategy targeting all newborns--regardless of month of birth--to effectively reduce viral circulation at the population level, largely independently of the capacity to implement catch-up immunization, while also simplifying delivery pathways and reducing overall programmatic costs.

Our study has some limitations that also suggest directions for future work. In particular: (i) the lack of complete age-specific epidemiological and immunisation data limits the possibility of calibrating the model to observed outcomes. While our analysis is designed to explore and compare different Nirsevimab administration strategies rather than to perform data fitting, more detailed information on infection rates, age-specific susceptibility, and antibody coverage would allow future studies to further refine parameter estimates and strengthen quantitative assessments. (ii) The study is limited to the Italian context. Specific quantitative outcomes may not be directly generalizable to other countries, although the qualitative insights and relative comparison of protection strategies are expected to remain broadly informative. (iii) The model is intentionally restricted to infant prophylaxis with Nirsevimab and does not include RSV vaccination in older adults. This choice allows us to isolate the effect of infant protection strategies, but it prevents the assessment of the direct and indirect contribution that vaccination of older adults may provide.

Future extensions could incorporate vaccination strategies for older adults, using a scenario-based approach analogous to that adopted here for infant prophylaxis, in order to assess their direct impact on the senior population and their possible indirect effects on overall RSV transmission. Similarly, in the case of maternal vaccination, it would be valuable to compare its protective effect with that of monoclonal antibodies to identify the most effective strategy for infant protection. Additionally, more detailed and consistent surveillance data would allow refinement of age-specific contact patterns, improving quantitative estimates while preserving the public health relevance of the proposed framework.

\paragraph*{Acknowledgements.} This work has been carried out under the auspices of the Italian National Group for Mathematical Physics (GNFM) of the National Institute for Advanced Mathematics (INdAM).  The research was supported by EU funding within the NextGenerationEU---MUR PNRR Extended Partnership initiative on Emerging Infectious Diseases (Project no. PE00000007, INF-ACT).

\paragraph*{Data availability statement.} Data sharing is not applicable to this article as no datasets were generated or
analysed during the current study.

\section*{Compliance with ethical standards}

\paragraph*{Conflict of interest.}
The authors state that there is no conflict of interest.

\appendix 
\section*{Appendix} 
\renewcommand{\thesection}{A.\arabic{section}}
\setcounter{section}{0}

\section{Proof of Proposition \ref{prop:invariance}}
By standard procedure (see e.g.~\cite{buonomo2025modelling}), one can derive the positive invariance of the strictly positive cone of $\R^{11}$. Since the closure of a positively invariant set is still positively invariant \cite[Remark 16.3h]{amann2011ordinary}, it follows that also the nonnegative cone $\R_+^{11}$ is positively invariant.

Adding the equations of model \eqref{eq: sys-am}, we obtain the following equation ruling the evolution of the total population $N(t)$:
\begin{align*}
    \dot N = \Lambda - \mu_1 N_1 - \mu_2 N_2 - \mu_3 N_3 -d_1 I_1 - d_3 I_3\\
    \leq \Lambda - \mu_1 N_1 - \mu_2 N_2 - \mu_3 N_3 \\
    \leq \Lambda - \min\{\mu_1, \mu_2, \mu_3\}N.
\end{align*}
Denoting as $\mu = \min\{\mu_1, \mu_2, \mu_3\}$, the comparison with the linear equation $\dot y=\Lambda-\mu y$ yields
\begin{equation*}
N(t)\le N(0)e^{-\mu t}+\frac{\Lambda}{\mu}\left(1-e^{-\mu t}\right),\qquad t\ge 0,
\end{equation*}
hence
\begin{equation*}
\limsup_{t\to\infty}N(t)\le \frac{\Lambda}{\mu}
\quad\text{and}\quad
N(t)\le \max\!\left\{N(0),\,\frac{\Lambda}{\mu}\right\},\quad t \geq 0.
\end{equation*}

\section{The \emph{next--generation} matrix method}
\label{sec: DFE}
To compute the control reproduction number, consider the next--generation matrix approach. We denote by $\mathcal{F}$ the vector of new infection rates into the infected compartments:
\begin{equation*}
    \mathcal{F}(\mathbf{x})= \begin{bmatrix}
			(\sigma P+S_{1e}+S_{1w})(\beta_{11}I_1 + \beta_{12}I_2) \\ S_2(\beta_{21}I_1 + \beta_{22}I_2) \\ S_3(\beta_{31}I_1+\beta_{32}I_2 + \beta_{33}I_3) \\
		\end{bmatrix}
\end{equation*}
and denote by $\mathcal{V}$ the vector of inflows and outflows from the infected compartments due to all processes other than new infections, namely recovery, disease-induced death, age transition, and natural death:
\begin{equation*}
\mathcal{V}(\mathbf{x}) = 
    \begin{bmatrix}
        \xi_1 I_1 \\
        \xi_2 I_2 - \eta_1I_1 \\
        \xi_3 I_3 - \eta_2I_2\\
    \end{bmatrix},
\end{equation*}
where the parameters $\xi_i$, $i=1,2,3$, are defined in \eqref{eq: xi_i}. The Jacobian of $\mathcal{F}$ and $\mathcal{V}$ evaluated at the disease--free equilibrium are the so--called \emph{transmission} matrix and the \emph{transition} matrix, respectively \cite{diekmann2010construction,van2002reproduction}:

\begin{equation*}
 \mathbf{F}=\begin{bmatrix}
        A_{11} & A_{12} & 0\\
        A_{21} & A_{22} & 0\\
        A_{31} & A_{32} & A_{33}
    \end{bmatrix},
    \qquad
   \mathbf{V}=\begin{bmatrix}
        \xi_1 & 0  &  0\\
        -\eta_1 & \xi_2 & 0\\
        0 & -\eta_2 & \xi_3  
    \end{bmatrix},
\end{equation*}
where: 
\begin{equation*} \label{eq: defAij_Appendix}
\begin{aligned}
    A_{11}&=\beta_{11}(\sigma P^0+S^0_{1e}+S^0_{1w}), & A_{12}&=\beta_{12}(\sigma P^0+S^0_{1e}+S^0_{1w}), \\ A_{21}&=\beta_{21}S^0_2, &
    A_{22}&=\beta_{22}S^0_2,\\ A_{31}&=\beta_{31}S_3^0, & A_{32}&=\beta_{32}S^0_3,& A_{33}&=\beta_{33} S^0_3.
\end{aligned}
\end{equation*}
The control reproduction number $\Rc$ is the spectral radius of the $\mathbf{NGM}$, given by:
\begin{equation*}
\mathbf{NGM}=\mathbf{FV}^{-1}=\begin{bmatrix}
        \mathbf{B} & 0 \\
        * & A_{33}\xi_{3}^{-1}
    \end{bmatrix},
\end{equation*}
where the entries $b_{ij}$ of the matrix $\mathbf{B}$ can be expressed as:
\begin{equation*}
    \begin{aligned}
        b_{11}&=\frac{A_{11}}{\xi_1}+\frac{A_{12}\eta_1}{\xi_1\xi_2}, & b_{12}&=\frac{A_{12}}{\xi_2}, \\
        b_{21}&=\frac{A_{21}}{\xi_1}+\frac{A_{22}\eta_1}{\xi_1\xi_2}, & b_{22}&=\frac{A_{22}}{\xi_2}.
    \end{aligned}
\end{equation*}
Thus, the first two eigenvalues of $\mathbf{FV}^{-1}$ correspond to the eigenvalues of matrix $\mathbf{B}$, while the third one is $\lambda_3=A_{33}\xi_3^{-1}$. This yields the control reproduction number
\begin{equation*}\label{eq: Rc}
    \Rc= \max \{\Rc^{(1,2)}, \Rc^{(3)} \},
\end{equation*}
where
\begin{equation*} 
    \Rc^{(1,2)} =\frac{1}{2} \left[ \frac{A_{11}}{\xi_1} + \frac{A_{12}\eta_1}{\xi_1\xi_2} + \frac{A_{22}}{\xi_2} + \sqrt{\left(\frac{A_{11}}{\xi_1} + \frac{A_{12}\eta_1}{\xi_1\xi_2} - \frac{A_{22}}{\xi_2}\right)^2 + 4\frac{A_{12}}{\xi_2} \left(\frac{A_{21}}{\xi_1} + \frac{A_{22}\eta_1}{\xi_1\xi_2}\right)}\;\right],
\end{equation*}
\begin{equation*}
    \Rc^{(3)} =\frac{\beta_{33}S_3^0}{\xi_3}.
\end{equation*}

When no control measures are implemented, the model admits the basic reproduction number $\Rzero$. Following the same procedure used to derive the control reproduction number $\Rc$, it is found that $\Rzero$ has exactly the same formal expression. The only difference is that the coefficients $A_{ij}$ are now evaluated at the disease-free equilibrium of the uncontrolled system, which is now given by 
\begin{equation*} 
    \bar E_0 = \left(0, \bar S_{1e}^0, 0, 0, 0, \bar S_2^0, 0, 0, \bar S_3^0, 0, 0\right),
\end{equation*}
and whose explicit coordinates are
\begin{equation} \label{eq: coord DFE_bar}
    \bar S_{1e}^0 = \frac{\Lambda}{\mu_1 + \eta_1}, \quad 
    \bar S_2^0 = \frac{\eta_1 \bar S_{1e}^0}{\mu_2 + \eta_2}, \quad 
    \bar S_3^0 = \frac{\eta_2 \bar S_2^0}{\mu_3}.
\end{equation}

\section{
Proof of Proposition~\ref{thm: LAS DFE} }\label{sec:LAS_DFE}

The Jacobian matrix $\mathbf{J}(E_0)$ evaluated at the disease-free equilibrium $E_0$ is
\begin{equation*}
\label{Jacobiano}
     \mathbf{J}(E_0)=\begin{bmatrix}
     \mathbf{B}_{8\times8} & \mathbf{0}_{8\times3}\\
     \mathbf{C}_{3\times8} & \mathbf{D}_{3\times3}
     \end{bmatrix},
\end{equation*}
where $\mathbf{B}$ is

\medskip

	\adjustbox{max width=0.90\textwidth}{$\displaystyle
     \begin{bmatrix}
        -(\mu_1+\omega)&\psi & 0 & -\sigma P^0 \beta_{11} & 0 & 0 & -\beta_{12} \sigma P^0 & 0 \\
        0 & -(\mu_1+\eta_1+\psi)&0 & -\beta_{11}S^0_{1e} & 0 & 0 & -\beta_{12}S^0_{1e}  & 0 \\ q\omega &0&-(\mu_1+\eta_1)&- \beta_{11}S^0_{1w}& \nu_1 &0& -\beta_{12}S^0_{1w}&0
         \\
         0 & 0 & 0&  A_{11} - \xi_1 &0 & 0 & A_{12} &  0  \\
         0 & 0 &0&  \gamma_1 & -(\mu_1+\eta_1+\nu_1) & 0 & 0 & 0  \\
         (1-q)\omega & \eta_1&\eta_1 & -\beta_{21}S^0_2 & 0 & -(\mu_2+\eta_2) & -A_{22} & \nu_2 \\
         0 & 0 &0&  A_{21} & 0 & 0 & A_{22}-\xi_2 & 0 \\
         0 & 0 & 0&0 & \eta_1 &0 & \gamma_2 & -(\mu_2+\eta_2+\nu_2) 
    \end{bmatrix},
    $}
    \bigskip
\begin{equation*}
     \mathbf{C}=\begin{bmatrix}
        
           0&0&0& -A_{31}&0&\eta_2 & -\beta_{32}S^0_3&0\\
           0&0&0&A_{31}&0&0 & A_{32}& 0\\
            0&0&0&0&0&0&0 & \eta_2 
    \end{bmatrix},
\end{equation*}
and
\begin{equation*}
\label{eq: matrix D}
     \mathbf{D}=\begin{bmatrix}
        
           -\mu_3 & -A_{33} &\nu_3\\
           0 &A_{33}-\xi_3 & 0\\
            0 & \gamma_3 & -(\mu_3+\nu_3)
    \end{bmatrix}.
\end{equation*}
Since $\mathbf{J}(E_0)$ is a block lower triangular matrix, its eigenvalues correspond to those of the submatrices $\mathbf{B}$ and $\mathbf{D}$.
The eigenvalues of the matrix $\mathbf{D}$ are trivially the diagonal elements, which are all negative when $\Rc<1$.

The first six eigenvalues of the matrix $\mathbf{B}$ are
\begin{gather*}
    \lambda_1 = -(\mu_1 + \omega), \quad \lambda_2 = -(\mu_1 + \eta_1), \quad \lambda_3 = -(\mu_1 + \eta_1 + \nu_1),\\
     \lambda_4 = -(\mu_2+\eta_2), \quad \lambda_5 = - (\mu_2 + \eta_2 + \nu_2), \quad \lambda_6=-(\mu_1+\eta_1+\psi),
\end{gather*}
while the remaining ones can be computed by solving the equation
\begin{equation*}
    (A_{11}-\xi_1 - \lambda)(A_{22}-\xi_2-\lambda) -A_{12}A_{21} =0,
\end{equation*}
that is:
\begin{equation*}
    \lambda^2 + \lambda\underbrace{(-A_{11}+\xi_1-A_{22}+\xi_2)}_b + \underbrace{(A_{11}-\xi_1)(A_{22}-\xi_2) - A_{12}A_{21}}_c=0
\end{equation*}
The roots of this characteristic equation have negative real parts if $b>0$ and $c>0$.

Let us prove that $b>0$.
By hypothesis, $\Rc^{(1,2)}<1$, that is
     \begin{align}
     \label{eq: diseq R0}
     1> \Rc^{(1,2)} &\geq\frac{1}{2} \left[ \frac{A_{11}}{\xi_1}+ \frac{A_{22}}{\xi_2} + \sqrt{ \Bigg(\frac{A_{11}}{\xi_1}-\frac{A_{22}}{\xi_2}\Bigg)^2 +4\frac{A_{21}}{\xi_1}\frac{A_{12}}{\xi_2}}\; \right] \\
     \notag
     &\geq \frac{1}{2} \left[ \frac{A_{11}}{\xi_1}+ \frac{A_{22}}{\xi_2} + \sqrt{ \Bigg(\frac{A_{11}}{\xi_1}-\frac{A_{22}}{\xi_2}\Bigg)^2}\; \right] \\
     \notag
     & \geq \max\ \left\{ \frac{A_{11}}{\xi_1},\frac{A_{22}}{\xi_2} \right \}
     \end{align}
from which 
\begin{equation*}
   \frac{ A_{11}}{\xi_1} < 1 \ ,
    \qquad
    \frac{A_{22}}{\xi_2}<1.
\end{equation*}
 It follows trivially from the latter inequalities that $b>0$ .\\
 In order to verify that $c>0$, we consider the inequality given by \eqref{eq: diseq R0}:
 \begin{equation*}
   \frac{1}{2} \left[ \frac{A_{11}}{\xi_1}+ \frac{A_{22}}{\xi_2} + \sqrt{ \Bigg(\frac{A_{11}}{\xi_1}-\frac{A_{22}}{\xi_2}\Bigg)^2 +4\frac{A_{21}}{\xi_1}\frac{A_{12}}{\xi_2}}\; \right]   <1,
 \end{equation*}
that is equivalent to
 \begin{equation*}
   \left[\sqrt{ \Bigg(\frac{A_{11}}{\xi_1}-\frac{A_{22}}{\xi_2}\Bigg)^2 +4\frac{A_{21}}{\xi_1}\frac{A_{12}}{\xi_2}}\; \right] ^2 <\left[ 2- \left( \frac{A_{11}}{\xi_1} + \frac{A_{22}}{\xi_2} \right) \right]^2,
 \end{equation*}
which in turn is equivalent to
 \begin{equation*}
    \Bigg(\frac{A_{11}}{\xi_1}-\frac{A_{22}}{\xi_2}\Bigg)^2 + 4\frac{A_{21}}{\xi_1}\frac{A_{12}}{\xi_2} < 4 + \Bigg(\frac{A_{11}}{\xi_1}+\frac{A_{22}}{\xi_2}\Bigg)^2 - 4 \Bigg(\frac{A_{11}}{\xi_1}+\frac{A_{22}}{\xi_2}\Bigg).
 \end{equation*}
Finally, we obtain the following inequality:
 \begin{equation*}
    1 + \frac{A_{11}}{\xi_1}\frac{A_{22}}{\xi_2}-\Bigg(\frac{A_{11}}{\xi_1}+\frac{A_{22}}{\xi_2} \Bigg)-\frac{ A_{12}}{\xi_2} \frac{A_{21}}{\xi_1}>0.
 \end{equation*}
By factoring out the term $\xi_1 \xi_2$ in c, using the latter inequality, we obtain:
\begin{align*}
   c&= \xi_1\xi_2 \Bigg[ \Bigg(\frac{A_{11}}{\xi_1} -1\Bigg) \Bigg(\frac{A_{22}}{\xi_2} -1\Bigg) - \frac{A_{12}}{\xi_2}  \frac{A_{21}}{\xi_1}\Bigg]\\
    &=\xi_1\xi_2 \Bigg[ \frac{A_{11}}{\xi_1} \frac{A_{22}}{\xi_2} +1 - \frac{A_{22}}{\xi_2} - \frac{A_{11}}{\xi_1} - \frac{A_{12}}{\xi_2}  \frac{A_{21}}{\xi_1}\Bigg] >0
\end{align*}

Thus, all the eigenvalues have negative real parts, implying that $E_0$ is locally asymptotically stable when $\Rc<1$. This completes the proof.

\section{
Auxiliary results for the proof of uniform persistence} \label{sec:permanence}
Let us consider the definition of $M_{\partial}$ given in Section~\ref{sec:rep_num}.
The following results hold.
\begin{lemma} \label{Lemma:M=Mpartial}
   Let $M=\{( P, S_{1e}, S_{1w},0,0,S_2,0,0,S_3,0,0)\in X\colon P, S_{1e}, S_{1w}, S_2, S_3 \geq 0\}$.\\
   Then $M = M_{\partial}$
\end{lemma}
\begin{proof}
    Firstly, note that $M \subseteq M_{\partial}$ holds trivially. Then, let us prove that $M_{\partial} \subseteq M$, that is, if $(P(0),S_{1e}(0), S_{1w}(0), I_1(0), R_1(0), S_2(0), I_2(0), R_2(0), S_3(0),I_3(0),R_3(0)) \in M_{\partial}$, then $I_i(0)=0$ for all $i=1,2,3$. By contradiction, assume $I_1(0)>0$, without loss of generality. For $ t \in [0, T_1]$, consider the differential equation for $I_1(t)$ from system~\eqref{eq: sys-am}:
    \begin{equation*}
        \dot I_1(t) = (\sigma P + S_{1e} + S_{1w})(\beta_{11} I_1 + \beta_{12} I_2) - (\mu_1  + \gamma_1 + d_1+ \eta_1) I_1.
    \end{equation*}
    By dropping the non-negative infection terms and bounding from below, we obtain the inequality
    \begin{equation} \label{eq:I1_bounded}
        \dot I_1(t) \geq - \xi_1 I_1(t),
    \end{equation}
    where $\xi_1 := \mu_1 + \gamma_1 + d_1 + \eta_1$.  
    Solving inequality~\eqref{eq:I1_bounded}, we find
    \[
    I_1(t) \geq I_1(0) e^{-\xi_1 t}.
    \]
    Since $I_1(0) > 0$ by assumption, it follows that $I_1(t) > 0$ for all $t \in [0, T_1]$, and in particular, there exists a constant $q_1 > 0$ such that
    \[
    I_1(t) \geq q_1 \quad \text{for all } t \in [0, T_1].
    \]
    Substituting $I_1$ into equation for $I_2(t)$ in \eqref{eq: sys-am} for $t \in [0,T_1]$ we obtain
    \begin{equation*}
        \dot I_2(t)\geq \eta_1q_1 - \xi_2 I_2(t),
    \end{equation*}
    and, then
    \begin{equation*}
        I_2(t) \geq \frac{\eta_1 q_1}{\xi_2} ( 1-e^{-\xi_2T_1})>0
    \end{equation*}
    Hence, there exists a constant  $q_2 > 0$ such that $I_2(t)\geq q_2$. Using this in the equation for $I_3(t)$ and proceeding analogously, we deduce that $I_3(t)>0$ on the same interval. Thus, whenever $I_1(0)>0$, both $I_2(t)$ and $I_3(t)$ remain positive for some time. By the definition of $M_{\partial}$, this is a contradiction. The same argument holds for all other points on $\partial X_0$ except for $(P, S_{1e}, S_{1w}, 0,0,S_2,0,0,S_3,0,0)$. 
\end{proof}

\begin{lemma}\label{Lemma:E0_weakrep}
    If $\Rc>1$, then the disease free equilibrium $E_0$ of system~\eqref{eq: sys-am} is weak repeller for $X_0$, i.e.
    \begin{equation*}
        \limsup_{t \to \infty}  \dist (\Phi(t), E_0) >0
    \end{equation*}
    where $\Phi(t)=(P(t), S_{1e}(t), S_{1w}(t), I_1(t), R_1(t),S_2(t), I_2(t), R_2(t),S_3(t), I_3(t), R_3(t)) $ is an arbitrary solution of system \eqref{eq: sys-am} with any initial value in $X_0$.
\end{lemma}
\begin{proof}
   According to Leenheer and Smith (see the proof of Lemma 3.5 in \cite{smith2003virus}), it is sufficient to prove that $W^s(E_0) \cap X_0=\emptyset$, where $W^s(E_0)$
denotes the stable manifold of $E_0$. Suppose, for contradiction, that the intersection is nonempty. Then there exists a solution $(P, S_{1e}, S_{1w}, I_1,R_1, S_2, I_2,R_2, S_3, I_3,R_3)$ in $X_0$ such that 
\begin{equation*}
\begin{aligned}
    P(t) &\to P^0,\quad S_{1e}(t)\to S^0_{1e}, \quad S_{1w}(t)\to S^0_{1w}, \quad S_2(t) \to S^0_2,
    \quad S_3(t) \to S^0_3, \\
    I_i(t) &\to 0, \quad R_i(t)\to 0 \quad \text{as } t\to\infty, \quad i=1,2,3
\end{aligned}
\end{equation*}
where $P^0$, $S_{1e}^0$, $S_{1w}^0$, $S_i^0$ with $i=2,3$
are the equilibrium values at $E_0$.  

By the definition of limit, for each $\varepsilon_j>0$ there exists $T_j>0$ such that the corresponding inequality holds for all $t > T_j$:
\[
\begin{aligned}
P(t) &\ge P^0 - \varepsilon_1, & 
S_{1e}(t) &\ge S_{1e}^0 - \varepsilon_2, & 
S_{1w}(t) &\ge S_{1w}^0 - \varepsilon_3, \\
S_2(t) &\ge S_2^0 - \varepsilon_4, & 
S_3(t) &\ge S_3^0 - \varepsilon_5, & 
I_i(t) &< \varepsilon_6, \quad i=1,2,3, \\
R_i(t) &\ge 0, \quad i=1,2,3.
\end{aligned}
\]

Let
\[
T_{\max} = \max_j T_j.
\]
Then for all $t > T_{\max}$ all the inequalities above hold simultaneously.  

We next consider the dynamics of the infectious compartments. For this purpose, we introduce new variables $z_i$ ($i=1,2,3$) corresponding to $I_i$ in system~\eqref{eq: sys-am}, and define the following comparison system, for $t > T_{\max}$:
\begin{subequations}
\allowdisplaybreaks
\label{eq:comp_sys}
   \begin{align*}
    \dot z_1&=(\beta_{11}z_1 + \beta_{12}z_2)[ \sigma (P^0 - \varepsilon_1)+(S^0_{1e} - \varepsilon_2) + (S^0_{1w}-\varepsilon_3)]-\xi_1 z_1, \\
    \dot z_2&=\eta_1z_1+ (S^0_2 - \varepsilon_4)(\beta_{21} z_1 + \beta_{22}z_2) - \xi_2 z_2,\\
    \dot z_3&=\eta_2z_2+ (S^0_3- \varepsilon_5)(\beta_{31}z_1+\beta_{32} z_2 + \beta_{33}z_3) - \xi_3 z_3.
    \end{align*}
\end{subequations}

The Jacobian matrix of this system at the disease-free equilibrium $E_0$ is
\begin{equation*}
    J=J_0 -(\sigma\varepsilon_1+\varepsilon_2+\varepsilon_3)Q_1 - \varepsilon_4 Q_2 -\varepsilon_5 Q_3,
\end{equation*}
where

    \begin{equation*}
        J_0=\begin{bmatrix}
        \beta_{11}(\sigma P^0 + S_{1e}^0 +S^0_{1w}) -\xi_1 & \beta_{12}(\sigma P^0 + S_{1e}^0+S_{1w}^0) & 0 \\
        \eta_1+S_2^0\beta_{21} & \beta_{22}S_2^0 -\xi_2 & 0\\
        S^0_3\beta_{31}& \eta_2+S_3^0\beta_{32} & S^0_3 \beta_{33} - \xi_3
        \end{bmatrix},
    \end{equation*}
    \begin{equation*}
        Q_1=\begin{bmatrix}
        \beta_{11} & \beta_{12} & 0\\
        0&0&0\\
        0&0&0
        \end{bmatrix},
        \quad 
        Q_2=\begin{bmatrix}
        0&0&0\\
        \beta_{21} & \beta_{22}& 0\\
        0&0&0
        \end{bmatrix},
        Q_3=\begin{bmatrix}
        0&0&0\\
        0&0&0\\
        \beta_{31}& \beta_{32} & \beta_{33}
        \end{bmatrix}.
    \end{equation*}
    Note that $J_0$ is the Jacobian matrix of system~\eqref{eq: sys-am} evaluated at $E_0$, and $Q_1$, $Q_2$, $Q_3$ are nonnegative matrices defined accordingly. 
    Since $R_C>1$, the spectral bound of $J_0$ satisfies
\[
s(J_0)=\max\{\Re(\lambda):\lambda\in\sigma(J_0)\}>0.
\]
Moreover, by continuity of the spectral bound with respect to the matrix entries, for $\varepsilon_1,\ldots,\varepsilon_5$ sufficiently small, the matrix $J$ still satisfies $s(J)>0$.
    Since $J$ is a quasi-positive matrix, the Perron–Frobenius theorem ensures the existence of a positive eigenvector $\zeta > 0$ such that $J \zeta = s(J) \zeta$. Consequently, the solution $(z_1(t), z_2(t), z_3(t))\to \infty \quad as \quad  t \to \infty$. By the comparison principle, this implies that $(I_1(t), I_2(t), I_3(t)) \to \infty \quad as \quad t \to \infty$, which leads to a contradiction. This completes the proof.
\end{proof}

\section{Senior-only infectious equilibrium} \label{sec:senior-only}

As discussed in Section~2, it is reasonable to assume that \( c_{13} = c_{23} = 0 \). Under this assumption, the system theoretically admits an additional boundary equilibrium in which infection persists exclusively in the senior age group. This equilibrium, referred to as the \emph{senior--only endemic equilibrium}, is given by
\begin{equation*}
    \hat{E}=(\hat{P}, \hat{S}_{1e}, \hat{S}_{1w} , 0,0,\hat{S}_2, 0,0,\hat{S}_3, \hat{I}_3, \hat{R}_3),
\end{equation*} 
where
\begin{equation*}
    \hat{P}=P^0,\quad \hat{S}_{1e}=S^0_{1e},\quad \hat{S}_{1w}=S^0_{1w},\quad \hat{S}_2=S^0_2,\quad \hat{S}_3=\frac{\xi_3}{\beta_{33}}, \quad \hat{R}_3= \frac{\gamma_3 \hat{I}_3 }{\nu_3 + \mu_3},
\end{equation*}
and
\begin{equation*}
    \hat{I}_3= \left[\frac{\mu_3+\nu_3}{\mu_3\xi_3+ \nu_3(\mu_3+d_3)}\right] \frac{\xi_3\mu_3}{\beta_{33}} (\mathcal{P}_C -1),\quad \text{with}  \quad \mathcal{P}_C=\frac{\eta_2 \hat{S}_2 \beta_{33}}{\mu_3 \xi_3}.
\end{equation*}
However, this equilibrium is biologically meaningful if and only if $P_C>1$ and is unstable for all the parameter values, as shown in the following. Note that:
\begin{align} \label{eq: P0 < R0}
    \mathcal{P}_C:=\frac{\eta_2 \hat{S}_2 \beta_{33}}{\mu_3 \xi_3} = \frac{\eta_2 S_2^0}{\mu_3} \frac{\beta_{33}}{\xi_3}
    =  S_3^0\frac{\beta_{33}}{\xi_3} 
    = \frac{A_{33}}{\xi_3} =\Rc^{(3)}< \Rc.
\end{align}
Then, we have that $\mathcal{P}_C<\Rc$.

\begin{prop}
    If $\mathcal{P}_C>1$, then the equilibrium $\hat{E}$ is unstable.  
\end{prop}
\begin{proof}
The Jacobian matrix evaluated at the equilibrium $\hat{E}$ has the block lower-triangular structure
 \begin{equation*}
     \mathbf{J}(\hat{E})=\begin{bmatrix}
         \mathbf{B}_{8\times8} & \mathbf{0}_{8\times3} \\ 
         \mathbf{C}_{3\times8} & \mathbf{D}_{3\times3} 
     \end{bmatrix}
 \end{equation*}
 where the submatrix $\mathbf{B}$ is the same as that defined in Proposition~\ref{thm: LAS DFE}. Since $\mathbf{J}(\hat{E})$ has a lower block--triangular structure, the eigenvalues of $\mathbf{J}(\hat{E})$ are those of the submatrices $\mathbf{B}$ and $\mathbf{D}$. We claim that if $\mathcal{P}_C>1$, the matrix $\mathbf{B}$ admits at least one eigenvalue with positive real part. 
 
 This follows from the inequality \eqref{eq: P0 < R0}, which guarantees that $\mathcal{P}_C>1$ implies $\Rc>1$. In turn, the condition  $\Rc>1$ implies that the Jacobian submatrix $\mathbf{B}$ has at least one eigenvalue with positive real part. Hence, under the condition $\mathcal{P}_C>1$ the Jacobian matrix $\mathbf{J}(\hat{E})$ has at least one eigenvalue with positive real part, and the equilibrium $\hat{E}$ is therefore unstable.
\end{proof}

\bibliographystyle{abbrv}
\bibliography{Bibliography}
\end{document}